\documentclass[10pt,reqno]{amsart}
\numberwithin{equation}{section}
\usepackage{verbatim}
\usepackage{amssymb}
\usepackage{enumerate, xspace}
\usepackage{marvosym} 
\usepackage[sans]{dsfont} 
\usepackage{bbm}
\usepackage{changebar}
\usepackage[colorlinks]{hyperref}

\usepackage[usenames,dvipsnames,svgnames,table]{xcolor}
\usepackage[charter]{mathdesign}

\usepackage{dsfont}
\usepackage{ulem}
\usepackage{setspace}

\newtheorem{thm}{Theorem}[section]
\newtheorem{lem}[thm]{Lemma}

\newtheorem{prop}[thm]{Proposition}

\newtheorem{definition}[thm]{Definition}

\newtheorem{rem}[thm]{Remark}

\newcommand\bZ{{\mathbb Z}}

\newcommand\ve{\varepsilon}

\newcommand\vf{\varphi}

\newcommand{\pat}[1]{\textcolor{black}{#1}}
\newcommand{\bjo}[1]{\textcolor{black}{#1}}

\newcommand\LL{{\mathbb L}}
\newcommand\EE{{\mathbb E}}
\newcommand\PP{{\mathbb P}}

\newcommand\RR{{\mathbb R}}

\newcommand\ls{\;{\lesssim}\,}

\newcommand{\mc}[1]{{\mathcal #1}}

\newcommand{\bb}[1]{{\mathbb #1}}

\begin{document}
\title[]{A microscopic model for a one parameter class of fractional Laplacians with Dirichlet boundary conditions}
\author{C.Bernardin}
\address{Universit\'e C\^ote d'Azur, CNRS, LJAD\\
Parc Valrose\\
06108 NICE Cedex 02, France}
\email{{\tt cbernard@unice.fr}}

\author{P.  Gon\c calves}
\address{Patr\'icia Gon\c calves Center for Mathematical Analysis,  Geometry and Dynamical Systems,
Instituto Superior T\'ecnico, Universidade de Lisboa,
Av. Rovisco Pais, 1049-001 Lisboa, Portugal and   Institut  Henri
Poincar\'e, UMS 839 (CNRS/UPMC), 11 rue Pierre et Marie Curie, 75231 Paris Cedex 05, France.}
\email{{\tt patricia.goncalves@math.tecnico.ulisboa.pt}}

\author{B. Jim\'enez Oviedo}
\address{Universit\'e C\^ote d'Azur, CNRS, LJAD\\
Parc Valrose\\
06108 NICE Cedex 02, France}
\email{{\tt byron@unice.fr}}

\thanks{}

\date{\today.}
\begin{abstract}
\bjo{ We prove the hydrodynamic limit for the symmetric exclusion process with long jumps given by a mean zero probability transition rate with infinite variance and  in contact with infinitely many reservoirs with density $\alpha$ at the left of the system and $\beta $ at the right of the system. The strength of the reservoirs is ruled by $\kappa N^{-\theta}>0$. Here $N$ is the size of the system, $\kappa>0$ and $\theta \in \RR$. Our results are valid for $\theta \leq 0$. For $\theta=0$, we obtain a collection of  fractional reaction-diffusion equations indexed by the parameter $\kappa$ and with Dirichlet boundary conditions. Their solutions also depend on $\kappa$. For $\theta<0$, the hydrodynamic equation corresponds to a reaction equation with Dirichlet boundary conditions. The case $\theta > 0$ is still open. For that reason we also analyze the convergence of the unique weak solution of the equation in the case $\theta =0$ when we send the parameter $\kappa$ to zero. Indeed, we conjecture that the limiting profile when  $\kappa\to 0$ is the one that we should obtain when taking small values of $\theta>0$.
}
\end{abstract}

\maketitle
 
\section{Introduction}
Normal (diffusive) transport phenomena are described by standard random walk models. Anomalous transport, in particular transport phenomena giving rise to superdiffusion, are nowadays encapsulated in the Lévy flights or Lévy walks framework \cite{DSU,DKZ} and appear in physics, finance, biology ... The term "Lévy flight" was coined by Mandelbrot and is nothing but a random walk in which the step-lengths have a probability distribution that is heavy tailed.  A (one-dimensional)  Lévy walker moves with a constant velocity $v$ for a heavy-tailed random time $\tau$ on a distance $x=v \tau$ in either direction with equal probability and then chooses a new direction and moves again. One then easily shows that for Lévy flights or Lévy walks, the space-time scaling limit $P(x,t)$ of the probability distribution  of the particle position $x(t)$ is solution of the fractional diffusion equation 
\begin{equation}
\label{eq:fle024}
\partial_t P = -c (-\Delta)^{\gamma/2} P 
\end{equation}
where $c$ is a constant and $\gamma \in (1,2)$. In physics, the description of anomalous transport phenomena by Lévy walks instead of Lévy flights is sometimes preferred despite the two models have the same scaling limit form provided by (\ref{eq:fle024}) because the first ones have a finite propagation of speed (see \cite{DKZ} for more details).

While Lévy walks and Lévy flights are today well known and popular models to describe superdiffusion in infinite systems in various application fields, there has been  recently several physical studies pointing out that it would be desirable to have a better understanding of Lévy walks in bounded domains. For bounded domains, boundary conditions and exchange with reservoirs or environment have to be taken into account. A particular interest for this problem is related to the description of anomalous diffusion of energy in low-dimensional lattices \cite{D,LLP} in contact with reservoirs \cite{DS,DSS,LP}. It is for example argued in \cite{LP} that the density profiles of Lévy walkers in a finite box with absorbtion-reflection-creation well reproduces the temperature profile of some chains of harmonic oscillators with conservative momentum-energy noise and thermostat boundaries. It is well established that superdiffusive systems are much more sensitive to the reservoirs and boundaries than diffusive systems but quantitative informations, like the form of the singularities of the profiles at the boundaries, are still missing. 

In this work, motivated by these studies, we propose a simple interacting particle system which may be considered as a substitute to Lévy flights in bounded domains with reservoirs when Lévy flights are moreover {\textit{interacting}}. Indeed, the previous studies consider only non-interacting cases. The system considered here is composed of interacting Lévy flights on a one-dimensional lattice. More exactly, the system is an exclusion process on a finite lattice of size $N$ with jumps having a distribution in the form $p(z) \sim |z|^{-(1+\gamma)}$, $1<\gamma<2$,  and which in contact with some reservoirs at density $\alpha$ (resp. $\beta$) at its left (resp. right boundary). The reservoirs coupling is modulated by a prefactor $\kappa N^{-\theta}$, $\kappa>0$, $\theta \in \RR$.  \pat{In this work we focus on the case $\theta\leq 0$ and  the case $\theta>0$ remains open.}

Our main result is the derivation of the hydrodynamic limit for the density of particles for this system. The limiting PDE depends {\footnote{In the diffusive case $\gamma>2$ the limiting PDE is given by the heat equation with Dirichlet boundary conditions \cite{BGJO}. It does nod depend of $\kappa$.}} on the value of $\kappa$ and takes the form of a fractional heat equation with a singular reaction term, \pat{see \eqref{eq:Dirichlet Equation}}. The singular reaction term fixes the density on the left to be $\alpha$ and on the right to be $\beta$. In our opinion this singular reaction term, which is due to the presence of the reservoirs, should be more considered as a boundary condition than as a reaction term. We obtain in this way a new family of regional fractional Laplacians on $[0,1]$ with zero Dirichlet boundary conditions indexed by $\kappa$ and taking the form
\begin{equation}\label{def:LLKappa}
{\bb L}_\kappa= \LL -\kappa V_1, \quad V_1 (u) = c_\gamma \gamma^{-1} (u^{-\gamma} +(1-u)^{-\gamma} ),
\end{equation}
where $c_\gamma$ is a constant depending on $\gamma$. These operators are symmetric non-positive when restricted to the set of smooth functions compactly supported in $(0,1)$. For $\kappa=1$, we recover the so-called restricted fractional Laplacian while in the limit $\kappa \to 0$ we get the so-called regional fractional Laplacian.  We recall that since the fractional Laplacian is a non-local operator, the definition of a fractional Laplacian with Dirichlet boundary conditions is not obvious from a modeling point of view. \bjo{In the PDE's literature several candidates have been proposed, for instance, "restricted fractional Laplacian", "spectral fractional Laplacian", "Neumann Fractional Laplacian " \cite{ BKO,Vaz}, but often without a clear physical interpretation}. A probabilistic interpretation of these operators is sometimes possible and may enlighten their meaning. The restricted fractional Laplacian ($\kappa=1$) corresponds to the generator of a $\gamma$-Lévy stable process killed outside of $(0,1)$, while the regional fractional Laplacian ($\kappa=0$) corresponds to the generator of a censored $\gamma$-Lévy stable process on $(0,1)$ \cite{censored,GM}. For $\kappa \ne 0,1$ we could rely on the Feynman-Kac formula but we do not pursue this issue here. \pat{As mentioned above our reservoirs are regulated by the parameters   $\kappa N^{-\theta}$, $\kappa>0$ and in this work  we focus on the case $\theta\leq 0$. The case $\theta> 0$ is quite interesting and we conjecture that 
for small values of $\theta>0$ it is given by \eqref{eq:Dirichlet Equation} for the choice $\kappa=0$. To support this conjecture, in Theorem \ref{convergence_rho^k_to_rho^0}, we analyse the convergence of the profile that we obtained for $\theta=0$ and which is indexed in $\kappa$, when $\kappa\to 0$ (we also analyse the case $\kappa\to\infty$ confirming the behaviour obtained from the microscopic system when $\theta < 0$) and  indeed, we obtain that the limiting profiles are weak solution of the conjectured equation. We remark that the main problem in analysing the behavior of the microscopic system in this case is at the level of the derivation of the Dirichlet boundary conditions, since the  two-blocks  estimate does not work. We leave this open problem for a future work.}
After having obtained the hydrodynamic limits, we have studied their stationary solutions $\bar \rho^\kappa$, which are not explicit apart from the case $\kappa=1$  and the case $\kappa=\infty$, i.e. $\bar \rho^\infty = \lim_{\kappa \to \infty} \bar \rho^\kappa$. These profiles coincide with the profiles of the microscopic system in their non-equilibrium stationary states (see \cite{BJ} for the $\kappa=1$ case). The bounded continuous function $\bar \rho^\kappa$ has $\alpha$ and $\beta$ as boundary conditions and is such that it solves in a distributional sense \pat{the equation}
\begin{equation}
\label{eq:mml}
\bb L_\kappa \bar \rho^\kappa =-\kappa V_0, \quad V_0 (u) = c_\gamma \gamma^{-1} (\alpha u^{-\gamma} + \beta (1-u)^{-\gamma} ).
\end{equation}
There are many recent studies focusing on the regularization properties of fractional operators in bounded domains. Even in this one dimensional setup, the question is in general non trivial.  For $\kappa=1$, $\bar\rho^\kappa$ can be computed explicitly and it appears that it is smooth in the interior of $[0,1]$ but has only H\"older regularity equal to $\gamma/2$ at the boundaries. For $\kappa\ne 1$, it should be possible to prove the interior regularity of $\bar \rho^\kappa$ by some existing methods (\cite{Mou}) but the boundary regularity that numerical simulations seem to indicate to depend on $\kappa$ is much more challenging and seems to be open. We prove that as $\kappa \to 0$, $\bar \rho^\kappa \to \bar \rho^0$ in a suitable topology and that $\bar \rho^0$ is a weakly harmonic function of the regional fractional Laplacian $\bb L_0$, i.e. we can take $\kappa=0$ in (\ref{eq:mml}). We left these interesting questions for future works.

The paper is organized as follows. In Section    \pat{\ref{sec:stat_results} we introduce the model  and we present all the PDE's that will be related to its hydrodynamic limit. We also present the main results of this work, namely the hydrodynamic limit stated in Theorem \ref{theo:hydro_limit}, the convergence, when $\kappa\to0$ and when $\kappa\to\infty$, of the hydrodynamical profile  in Theorem \ref{convergence_rho^k_to_rho^0} and of the stationary profile in Theorem \ref{Existence_uniqueness_convergence_stationary}.
Section \ref{sec:proof_hyd_limit}
is devoted to the proof of Theorem \ref{theo:hydro_limit} while Sections \ref{sec: Study of solution} and \ref{sec:Exis_unq_sta_sol} are dedicated, respectively,  to the convergence of the hydrodynamical profile and of the stationary profile. Finally, in Section \ref{sec:Uniqueness} we prove the uniqueness of all the weak solutions that we consider in this work.}
\section{Statement of results}\label{sec:stat_results}

\subsection{The model}
\label{sec:model}
For $N\geq{2}$ let  $\Lambda_N=\{1, \ldots, N-1\}$.  The boundary driven exclusion process with long jumps is a  Markov process that we denote by $\{\eta(t)\}_{t\geq{0}}$ with state space $\Omega_N:=\{0,1\}^{\Lambda_N}$ and is defined as follows. The configurations of the state space $\Omega_N$  are denoted by $\eta$, so that for $x\in\Lambda_N$,  $\eta_{x}=0$ means that the site $x$ is vacant while $\eta_{x}=1$ means that the site $x$ is occupied. Fix $\gamma \in (1,2)$. Let $p:\mathbb{Z}\rightarrow{[0,1]}$ be a translation invariant transition probability defined by
\begin{equation}\label{transition_p}
p(z) = c_{\gamma}\dfrac{{\bb 1}_{\{z \ne 0\}}}{\vert z\vert^{\gamma+1}}\, 
\end{equation}
where  $c_{\gamma}$ is a normalizing constant. Since $\gamma\in (1,2)$, we know that $p$ has infinite variance but finite mean.  

Fix $0<\alpha\leq\beta<1$. We consider the process in contact with infinitely many stochastic  reservoirs with density $\alpha$ at all the negative integer sites  and  with density ${\beta}$ at all the integer sites $z\geq N$. The intensity of the reservoirs is regulated by a parameter $\kappa N^{-\theta}$ where $\kappa>0$ and $\theta\leq 0$.

 The  process is characterized by its infinitesimal generator \begin{equation}
\label{Generator}
 L_{N} = L_{N}^{0}+\kappa N^{-\theta} L_{N}^{\ell}+ \kappa N^{-\theta} L_{N}^{r},
\end{equation}
which acts on functions  $f:\Omega_N \to \RR$ as
\begin{equation}\label{generators}
\begin{split}
&(L^0_N f)(\eta) =\cfrac{1}{2} \, \sum_{x,y \in \Lambda_N} p(x-y) [ f(\sigma^{x,y}\eta) -f(\eta)],\\
&(L_N^{\ell} f)(\eta) =\sum_{\substack{x \in \Lambda_N\\ y \le 0}} p(x-y)c_{x}(\eta;\alpha) [f(\sigma^x\eta) - f(\eta)],\\
&(L_N^{r} f)(\eta)= \sum_{\substack{x \in \Lambda_N \\ y \ge N}} p(x-y) c_{x}(\eta;\beta)  [f(\sigma^x\eta) - f(\eta)]
\end{split}
\end{equation}
where 
\begin{equation*}
(\sigma^{x,y}\eta)_z = 
\begin{cases}
\eta_z,& \textrm{if}\;\; z \ne x,y,\\
\eta_y,& \textrm{if}\;\; z=x,\\
\eta_x,& \textrm{if}\;\; z=y
\end{cases}
, \quad (\sigma^x\eta)_z= 
\begin{cases}
\;\; \eta_z, &\textrm{if}\;\; z \ne x,\\
1-\eta_x,& \textrm{if}\;\; z=x,
\end{cases}
\end{equation*}
and  for a function $\vf:[0,1]\rightarrow \RR$ and for $x\in \Lambda_{N}$ we used the notation 
\begin{equation}\label{rate_c}
c_{x} (\eta;\vf(\cdot)) :=\left[ \eta_x  \left(1-\vf(\tfrac{x}{N}) \right) + (1-\eta_x)\vf(\tfrac{x}{N})\right].
\end{equation}

We consider  the Markov process  speeded up in the subdiffusive time scale $t\Theta(N)$ and we use the notation $\eta_{t}^{N}:= \eta(t\Theta(N))$, so that $\eta_t^N$ has  infinitesimal generator $\Theta(N) L_{N}$. Although $\eta _{t}^{N}$ depends on $\alpha$, $\beta$ $\theta$ and $\kappa$, we shall omit these indexes in order to simplify notation. 

\subsection{Hydrodynamic equations}
\label{subsec:hyd_eq}
From now on up to the rest of this article we fix a finite time horizon $[0,T]$. 
To properly  state the hydrodynamic limit, we need to introduce some notations and definitions, which we present as follows: {first} we abbreviate the Hilbert space $L^{2}([0,1],h(u)du)$ by $L^{2}_{h}$ and we denote  its inner product by $\langle \cdot,\cdot\rangle _{h}$ and the corresponding norm  by $\Vert\cdot \Vert_{h}$. When $h\equiv 1$  we simply write $L^{2}$, $\langle \cdot,\cdot\rangle$ and $\Vert\cdot\Vert$. For  an interval $I$ in $\RR$ and integers
$m$ and $n$, we denote by $C^{m,n}([0, T] \times  I)$ the set of functions defined on $[0, T] \times I $ that are $m$ times differentiable on the first variable and $n$ times differentiable  on the second variable. We denote by $C_{c}^{\infty}( I)$ the set of all smooth real-valued  functions defined in $ I$ with compact support included in $ I$. The supremum norm is denoted by $\Vert \cdot\Vert_{\infty}$. We also consider the set $C_c^{1,\infty} ([0,T]\times  I)$ of functions $G \in C^{1,\infty}([0, T] \times I)$ such that $G(t,\cdot) \in C_{c}^{\infty} (I)$ for all $t \in [0, T]$. An index on a function will always denote a
variable, not a derivative. For example, $G_{t}(u)$ means $G(t, u)$. The derivative of $G \in C^{m,n}([0, T] \times I)$ will be denoted by $\partial _{{t}}G$ (first variable) and $\partial _{u}G$ (second variable). 

The fractional Laplacian $-(-\Delta)^{\gamma/2}$ of exponent $\gamma/2$  is defined on the set of functions $G:\RR \to \RR$ such that
\begin{equation}
\label{eq:integ1}
\int_{-\infty}^{\infty} \cfrac{|G(u)|}{(1 +|u|)^{1+\gamma}} du < \infty
\end{equation}
by
\begin{equation}
-(-\Delta)^{\gamma/2} G \, (u) = c_\gamma  \lim_{\ve \to 0} \int_{-\infty}^{\infty} {\bb 1}_{|u-v| \ge \ve} \, \cfrac{G(v) -G(u)}{|u-v|^{1+\gamma}} dv
\end{equation}
provided the limit exists (which is the case, for example, if $G$ is in the Schwartz space) and {where $c_{\gamma}$ is set in (\ref{transition_p})}. Up to a multiplicative constant, $-(-\Delta)^{\gamma/2}$ is the generator of a $\gamma$-L\'evy stable process. 

{We define the operator} $\bb L$ by its action on functions \pat{$G \in C_c^{\infty} ((0,1))$,} by
$$\forall u \in (0,1), \quad ({\bb L} G)(u) = c_\gamma  \lim_{\ve \to 0} \int_{0}^{1} {\bb 1}_{|u-v| \ge \ve} \, \cfrac{G(v) -G(u)}{|u-v|^{1+\gamma}} dv.$$
The operator ${\bb L}$ is called the \textit{regional fractional Laplacian} on $(0,1)$. The semi inner-product $\langle \cdot, \cdot \rangle_{\gamma/2}$ is defined on the set \pat{$C_c^{\infty}((0,1))$} by 
\begin{equation}
\langle G, H \rangle_{\gamma/2} =  \cfrac{c_{\gamma}}{2} \iint_{[0,1]^2} \cfrac{(H(u) -H(v)) (G(u) -G(v))}{|u-v|^{1+\gamma}} \, du dv.
\end{equation}  
The corresponding semi-norm is denoted by $\| \cdot \|_{\gamma/2}$. Observe that for any $G,H \in C_c^{\infty} ((0,1))$ we have that
\begin{equation*}
\langle G, {-\bb L} H \rangle = \langle {-\bb L} G, H \rangle = \langle G, H \rangle_{\gamma/2}.
\end{equation*}
 Recall \eqref{def:LLKappa}. We  introduced a family of  operators indexed by $\kappa$ and taking the form
\begin{equation*}
{\bb L}_\kappa= \LL -\kappa V_1.
\end{equation*}
Acting on $ C_{c}^{\infty}((0,1))$ these operators are symmetric and non-positive. For $\kappa=1$, we recover the so-called restricted fractional Laplacian (see \cite{Vaz}):
\begin{equation}\label{Operator_LL}
\forall u \in (0,1), \quad -(-\Delta)^{\gamma/2}  G\, (u)  =({\bb L} G)(u) -  V_1(u) G(u):=(\LL_{1}G)(u),
\end{equation}
 while in the limit $\kappa \to 0$ we get the  regional fractional Laplacian. 

We rewrite  $V_1(u)=r^-(u)+r^+(u)$ and $V_0(u)= \alpha r^-(u)+\beta r^+(u)$  where the functions $r^{\pm}: (0,1) \to (0, \infty)$ are defined by
\begin{equation} \label{def:rpm}
r^- (u)=c_\gamma \gamma^{-1} u^{-\gamma},\quad  r^+ (u) = c_\gamma \gamma^{-1} (1-u)^{-\gamma}.
\end{equation}

\begin{definition}
\label{Def. Sobolev space}
The Sobolev space $\mathcal{H}^{\gamma/2}:=\mathcal{H}^{\gamma/2}([0,1])$ consists of all square integrable functions $g: (0,1) \rightarrow \RR$ such that $\| g \|_{\gamma/2} <\infty$. This is a Hilbert space for the norm $\| \cdot\|_{{\mc H}^{\gamma/2}}$ defined by
$$\Vert g \Vert_{\mathcal{H}^{\gamma/2}}^{2}:= \Vert g \Vert^{2} + \Vert g \Vert_{\gamma/2}^{2} .$$
Its elements elements coincide a.e. with continuous functions. The completion of $C_c^{\infty} ((0,1))$ for this norm is denoted by ${\mc H}_0^{\gamma/2}:={\mc H}_0^{\gamma/2}([0,1])$. This is a Hilbert space whose elements coincide a.e. with continuous functions vanishing at $0$ and $1$. On ${\mc H}_0^{\gamma/2}$, the two norms $ \| \cdot \|_{{\mc H}^{\gamma/2}}$ and  $\| \cdot \|_{\gamma/2}$ are equivalent.   

The space $L^{2}(0,T;\mathcal{H}^{\gamma/2})$ is the set of measurable functions $f:[0,T]\rightarrow  \mathcal{H}^{\gamma/2}$ such that 
$$\int^{T}_{0} \Vert f_{t} \Vert^{2}_{\mathcal{H}^{\gamma/2}}dt< \infty. $$
The spaces $L^{2}(0,T;\mathcal{H}_0^{\gamma/2})$ and $L^{2}(0,T;L^{2}_h)$  are defined similarly.
\end{definition}

We now extend the definition of the regional fractional Laplacian on $(0,1)$, which has been defined on $C^{\infty}((0,1))$, to the space $\mc H ^{\gamma/2}$.

\begin{definition}\label{def:Dist}
For $\rho \in \mc H^{\gamma/2}$ we define the distribution $\LL \rho$ by $$\langle \LL\rho,G\rangle = \langle \rho,\LL G\rangle,\quad  G\in C_{c}^{\infty}((0,1)). $$
\end{definition}

Let us check that $\LL\rho$ is indeed a well defined distribution. Consider a sequence $\{G_{n}\}_{n\geq 1} \in C_{c}^{\infty}((0,1))$ converging to $0$  in the usual topology of the test functions. By the integration by parts formula for the regional fractional Laplacian (see Theorem 3.3 in \cite{GM}) we have for any $\rho\in \mc H^{\gamma/2}$ that $ \langle \LL\rho,G_{n}\rangle = \langle \rho, G_{n}\rangle_{\gamma/2}$. Now using the Cauchy-Schwarz's inequality and the mean value Theorem, we get that  $\langle \LL\rho,G_{n}\rangle$ is bounded from above by a constant times
$$ \Vert \rho \Vert_{\gamma/2}\Vert G_{n} \Vert_{\gamma/2}\ls \Vert \rho \Vert_{\gamma/2}\Vert G'_{n} \Vert_{\infty}^{2}\iint_{[0,1]^{2}} \vert u -v \vert^{1-\gamma}dudv $$
which goes to $0$ as $n\to \infty$ since $\gamma\in (1,2)$. Therefore $\LL\rho$ is a well defined distribution. 

\bjo{Above (and hereinafter) we write $f(u) \lesssim g(u)$ if there exists a constant $C$ independent of $u$ such that $f(u) \le C g(u)$ for every $u$. We will also write $f(u) = {O} (g(u) )$ if the condition $|f (u) | \lesssim |g(u) |$ is satisfied. Sometimes, in order to stress the dependence of a constant $C$ on some parameter $a$, we write $C(a)$.}
 
\subsection{Hydrodynamic equations}
\label{subsec:hyd_eq}
Now, for the following definitions recall the definition of $\LL_{\kappa}$ given in \eqref{def:LLKappa} and $V_{0}$ from \eqref{eq:mml}.

\begin{definition}
\label{Def. Dirichlet Condition}
 Let $\hat\kappa\geq 0$ be some parameter and let $g:[0,1]\rightarrow [0,1]$ be a measurable function. We say that  $\rho^{\hat\kappa}:[0,T]\times[0,1] \to [0,1]$ is a weak solution of the {non-homogeneous regional} fractional reaction-diffusion equation with Dirichlet boundary conditions given by  
\begin{equation}\label{eq:Dirichlet Equation}
\begin{cases}
&\partial_{t} \rho_{t}^{\hat\kappa}(u)= \LL_{\hat\kappa} \rho_t^{\hat\kappa}(u)+\hat\kappa V_0(u),  \quad (t,u) \in [0,T]\times(0,1),\\
 &{ \rho^{\hat\kappa}_{t}}(0)=\alpha, \quad { \rho^{\hat\kappa}_{t}}(1)=\beta,\quad t \in [0,T], \\
 &{ \rho}_{0}^{\hat\kappa}(u)= g(u),\quad u \in (0,1),
 \end{cases}
 \end{equation}
 if : 
\begin{enumerate}[i)] 

\item  $\rho^{\hat\kappa} \in L^{2}(0,T;\mathcal{H}^{\gamma/2})$.
\item   $\int_0^T \int_0^1 \Big\{ \frac{(\alpha-\rho_t^{\hat\kappa}(u))^2}{u^\gamma}+\frac{(\beta-\rho_t^{\hat\kappa}(u))^2}{(1-u)^\gamma}\Big\} \, du\, dt <\infty$ for $\hat\kappa>0$;   $\rho_t^{\hat\kappa}(0)=\alpha$, $\rho_t^{\hat\kappa}(1)=\beta$ for almost every $t\in[0,T]$,   for  $\hat\kappa=0$. 
\item For all $t\in [0,T]$ and {all functions} $G \in C_c^{1,\infty} ([0,T]\times (0,1))$ we have that 
\begin{equation}
\label{eq:Dirichlet integral}
\begin{split}
F_{Dir}(t, \rho^{\hat\kappa},G,g):=&\left\langle \rho_{t}^{\hat\kappa},  G_{t} \right\rangle -\left\langle g,   G_{0}\right\rangle - \int_0^t\left\langle \rho_{s}^{\hat\kappa},\Big(\partial_s + \bb L_{\hat\kappa} \Big) G_{s}  \right\rangle ds 
-\hat\kappa \int^{t}_{0}\left\langle G_s , V_0 \right\rangle\,ds=0.
\end{split}   
\end{equation}
\end{enumerate}
\end{definition}

\begin{rem}\label{use:rem_dir}
Note that item ii) is different for $\hat\kappa>0$ and $\hat\kappa=0$. We can see that the condition for $\hat\kappa=0$ is weaker than the condition for $\hat\kappa>0$. In fact, item i) and item ii) for $\hat\kappa>0$ of the previous definition imply that $\rho_t^{\hat\kappa} (0)=\alpha $ and $\rho_{t}^{\hat\kappa} (1) =\beta $, for almost every $t$ in $[0,T]$. Indeed, first note that by item i) we know that $\rho_t$  is  $\tfrac{\gamma-1}{2}$-H\"older for almost every $t$ in $[0,T]$ (see Theorem 8.2 of \cite{dPV} ). Then, we note that 
\begin{equation*}
 \int_{0}^{T} \dfrac{(\rho_{t}^{\hat\kappa}(0) - \alpha)^{2}}{\gamma-1}dt =\int_{0}^{T}\lim _{\ve\to 0}\ve ^{\gamma-1}\int_{\ve}^{1} \dfrac{(\rho_{t}^{\hat\kappa}(0) - \alpha)^{2}}{u^{\gamma}}dudt. 
\end{equation*}
By summing and subtracting $\rho_t^{\hat\kappa}(u)$ inside the square in the expression on the right hand side in   the previous equality and using the inequality $(a+b)^2\leq 2a^{2} +2b^{2}$ we get that the right hand side of  the previous equality  is bounded from above by
\begin{equation*}
\begin{split}
& 2\int_{0}^{T}\lim _{\ve\to 0}\ve ^{\gamma-1}\int_{\ve}^{1} \dfrac{(\rho_{t}^{\hat\kappa}(0) - \rho^{\hat\kappa}_{t}(u))^{2}}{u^{\gamma}}dudt 
+2\int_{0}^{T}\lim _{\ve\to 0}\ve ^{\gamma-1}\int_{\ve}^{1} \dfrac{(\rho_{t}^{\hat\kappa}(u) - \alpha)^{2}}{u^{\gamma}}dudt.
\end{split}
\end{equation*}
Since $\rho_t$  is  $\tfrac{\gamma-1}{2}$-H\"older for almost every $t$ in $[0,T]$  the term on the left hand side in the previous expression vanishes. Now, the term on the right hand side in the previous expression is bounded from above by
$$2\lim _{\ve\to 0}\ve ^{\gamma-1}\int_{0}^{T}\int_{0}^{1} \dfrac{(\rho_{t}^{\hat\kappa}(u) - \alpha)^{2}}{u^{\gamma}}dudt,$$
which vanishes {as a consequence of} item ii).
 Thus, we have that 
 $$\int_{0}^{T} \dfrac{(\rho_{t}^{\hat\kappa}(0) - \alpha)^{2}}{\gamma-1}dt =0,$$
 whence we get that  $\rho_t^{\hat\kappa} (0)=\alpha $ for almost every $t$ in $[0,T]$. Showing that $\rho_t^{\hat\kappa}(1)=\beta$ for almost every $t$ in  $[0,T]$ is completely analogous.
 
  Moreover, the existence and uniqueness {of a weak solution to the equation above,} for $\hat\kappa>0$ does not require the strong form of ii). Nevertheless, in order to prove Theorem \ref{convergence_rho^k_to_rho^0} we need {to impose} that condition.
\end{rem}

\begin{rem}
Observe that in the case  $\hat\kappa=1$, since $\LL_{1} =-(-\Delta)^{\gamma/2}$ we obtain {in Definition \ref{Def. Dirichlet Condition}} the fractional heat equation with reaction  and  Dirichlet boundary conditions, i.e.
  \begin{equation*}
 \begin{cases}
 &\partial_{t} \rho_{t}^{1}(u)= \LL_{1}\rho_t^{1}(u) + V_{0}(u), \quad (t,u) \in [0,T]\times(0,1),\\
 &{ \rho} _{t}^{1}(0)=\alpha, \quad { \rho}_{t}^{1}(1)=\beta,\quad t \in [0,T], \\
 &{ \rho}_{0}^{1}(u)= g(u),\quad u \in (0,1),
 \end{cases}
 \end{equation*}
  by \eqref{Operator_LL} and \eqref{def:LLKappa} the notion {of item iii) is reduced to} 
 \begin{equation*}
F_{Dir}(t, \rho^{1},G,g):=\left\langle \rho_{t}^{1},  G_{t} \right\rangle  -\left\langle g, G_{0} \right\rangle 
- \int_0^t\left\langle \rho_{s}^{1}, \Big(\partial_s -(-\Delta)^{\gamma/2} \Big) G_{s} \right\rangle ds - \int^{t}_{0}\left\langle G_s , V_0\right\rangle ds=0,   
 \end{equation*}
{for all $t\in [0,T]$ and all functions $G \in C_c^{1,\infty} ([0,T]\times (0,1))$.}
\end{rem}

\begin{definition}
\label{Def. Dirichlet Condition_kappa^infty}
 Let $\hat \kappa > 0$ be some parameter and let  $g:[0,1]\rightarrow [0,1]$ be a measurable function. We say that  $\rho^{\hat\kappa}:[0,T]\times[0,1] \to [0,1]$ is a weak solution of the {non-homogeneous} reaction equation with Dirichlet boundary conditions given by  
 
\begin{equation}\label{eq:Dirichlet Equation_infty}
 \begin{cases}
 &\partial_{t} \rho_{t}^{\hat \kappa}(u)= -\hat\kappa\rho^{\hat\kappa}_{t}(u)V_{1}(u) +\hat\kappa V_{0}(u),  \quad (t,u) \in [0,T]\times(0,1),\\
 &{ \rho} _{t}^{\hat\kappa}(0)=\alpha, \quad { \rho}_{t}^{\hat\kappa}(1)=\beta,\quad t \in [0,T], \\
 &{ \rho}_{0}^{\hat\kappa}(u)= g(u), \quad u \in (0,1),
 \end{cases}
 \end{equation}
if: 
\begin{enumerate}[i)] 

\item   $\int_0^T \int_0^1 \Big\{ \frac{(\alpha-\rho_t^{\hat\kappa}(u))^2}{u^\gamma}+\frac{(\beta-\rho_t^{\hat\kappa}(u))^2}{(1-u)^\gamma}\Big\} \, du\, dt <\infty$.
  
\item For all $t\in [0,T]$ and {all functions} $G \in C_c^{1,\infty} ([0,T]\times (0,1))$ we have
\begin{equation}\label{eq:Dirichlet integral_infty}
\begin{split}
F_{Reac}(t, \rho^{\hat\kappa},G,g):=&\left\langle \rho_{t}^{\hat\kappa},  G_{t} \right\rangle  -\left\langle g,   G_{0}\right\rangle- \int_0^t\left\langle \rho_{s}^{\hat\kappa},\partial_s G_{s}  \right\rangle ds
\\
&+ \int^{t}_{0} \left\langle \rho_{s}^{\hat\kappa},G_s \right\rangle_{V_1} ds - \int^{t}_{0}\left\langle  G_s ,V_0\right\rangle ds=0.  
\end{split}
\end{equation}
\end{enumerate}
\end{definition}
\begin{rem}\label{rem:Explicit_sol_rho_infty}
Note that the explicit solution of  (\ref{eq:Dirichlet Equation_infty}) is given by 
$$ \bar\rho^{\infty}(u) + (g(u)- \bar\rho^{\infty}(u))e^{-t\hat\kappa V_{1}(u)},$$
where $\bar\rho^{\infty}(u) = \dfrac{V_{0}(u)}{V_{1}(u)}$. As we will see, the function $ \bar\rho^{\infty}$ plays an important role {in the proof of our main results, namely,}  Theorems \ref{convergence_rho^k_to_rho^0} and \ref{Existence_uniqueness_convergence_stationary}.  
\end{rem}
{
\begin{lem}\label{lem:uniquess}
The weak solutions of (\ref{eq:Dirichlet Equation}) and (\ref{eq:Dirichlet Equation_infty}) are unique.
\end{lem}
Aiming to concentrate in the main facts, the proof of previous lemma is postponed to Section \ref{sec:Uniqueness}.}

\begin{definition}
\label{Def. Stationary_RFRD}
 Let $\hat\kappa \geq 0$ be some parameter. We say that  $\bar\rho^{\hat\kappa}:[0,1] \to [0,1]$ is a weak solution of the stationary regional fractional reaction-diffusion equation with {non-homogeneous} Dirichlet boundary conditions given by  
  \begin{equation}\label{eq:Stationary_RFRD}
 \begin{cases}
 & \LL_{\hat\kappa}\bar\rho^{\hat\kappa}(u)+\kappa V_0(u) = 0,  \quad  u \in (0,1),\\
 &{\bar\rho^{\hat\kappa}} (0)=\alpha, \quad { \bar\rho^{\hat\kappa}}(1)=\beta, 
 \end{cases}
 \end{equation}
if: 
\begin{enumerate}[i)] 
\item  $\bar\rho^{\hat\kappa} \in \mathcal{H}^{\gamma/2}$.
\item $\int^{1}_{0}\left\lbrace\tfrac{\left(\alpha-\bar\rho^{\hat\kappa}(u)\right)^{2}}{u^{\gamma}}+\tfrac{\left(\beta-\bar\rho^{\hat\kappa}(u)\right)^{2}}{u^{\gamma}} \right\rbrace du <\infty$ if $\hat\kappa>0$ and  ${\bar\rho^{\hat\kappa}} (0)=\alpha$, ${ \bar\rho^{\hat\kappa}}(1)=\beta$ if $\hat\kappa = 0$.
\item For any function $G \in C_c^{\infty} ((0,1))$ we have
\begin{equation*}
\bar F_{Dir}(\bar \rho^{\hat\kappa},G):=
\left\langle \bar\rho^{\hat\kappa}, \bb L_{\hat\kappa}  G \right\rangle 
 +\hat\kappa \left\langle G, V_0 \right\rangle =0.  
\end{equation*}
\end{enumerate}
\end{definition}

\begin{rem}
\pat{We observe  that $\bar\rho^{0}$ is a weak harmonic function for $\LL$ and}  the interior regularity of this solution is studied in \cite{Mou}, but the regularity at the boundary is unknown.
\end{rem}
In Section \ref{sec:Uniqueness}  we will prove the following lemma.
{
\begin{lem}
\label{existence and uniqueness}
There exists a unique weak solution of (\ref{eq:Stationary_RFRD}).
\end{lem}
}

\subsection{{Statement of results}}\label{subsec:HL}
  First we want to state the hydrodynamic limit of the process $\lbrace{\eta_{t}^{N}}\rbrace_{t\geq 0}$ with state space $\Omega_{N}$ and with infinitesimal generator $\Theta(N)L_{N}$ defined in (\ref{Generator}).

Let ${\mc M}^+$ be the space of positive measures on $[0,1]$ with total mass bounded by $1$ equipped with the weak topology. For any configuration  $\eta \in \Omega_{N}$ we define the empirical measure \pat{$\pi^{N}(\eta,du):=\pi^{N,\kappa}(\eta,du)$} in $\Omega_{N}$ by 
\begin{equation}\label{MedEmp}
\pi^{N}(\eta, du)=\dfrac{1}{N-1}\sum _{x\in \Lambda_{N}}\eta_{x}\delta_{\frac{x}{N}}\left( du\right),
 \end{equation}
where $\delta_{a}$ is a Dirac mass at $a \in [0,1]$ and
$\pi^{N}_{t}(\eta, du):=\pi^{N}(\eta_{t}^{N}, du).$

Let $g: [0,1]\rightarrow[0,1]$ be a measurable function. We say that a sequence of probability measures $\lbrace\mu_{N}\rbrace_{N\geq 1 }$ in $\Omega_{N}$  is associated to the profile $g$ if for any continuous function $G:[0,1]\rightarrow \mathbb{R}$  and every $\delta > 0$ 
\begin{equation*}
  \lim _{N\to\infty } \mu _{N}\left( \eta \in \Omega_{N} : \left\vert \dfrac{1}{N}\sum_{x \in \Lambda_{N} }G\left(\tfrac{x}{N} \right)\eta_{x} - \int_{0}^1G(u)g(u)du \right\vert    > \delta \right)= 0.
\end{equation*}

We denote by $\PP _{\mu _{N}}$ the probability measure in the Skorohod space $\mathcal D([0,T], \Omega_N)$ induced by the  Markov process ${\eta_{t}^{N} }$ and the measure $\mu_N$ in $\Omega_{N}$ and we denote by $\EE _{\mu _{N}}$ the expectation with respect to $\PP_{\mu _{N}}$. Let $\lbrace\mathbb{Q}_{N}\rbrace_{N\geq 1}$ be the  sequence of probability measures on the Skorohod space $\mathcal D([0,T],\mathcal{M}^{+})$ induced by the  Markov process $\lbrace \pi_{t}^{N}\rbrace_{t\geq 0}$ and by $\mathbb{P}_{\mu_{N}}$.

At this point we are ready to state the hydrodynamic limit of the process ${\eta_{t}^{N}}$.
\begin{thm}(Hydrodynamic limit)\label{theo:hydro_limit}
 Let $g:[0,1]\rightarrow[0,1]$ be a measurable function and let $\lbrace\mu _{N}\rbrace_{N\geq 1}$ be a sequence of probability measures in $\Omega_{N}$ associated to $g$. Then, for any $0\leq t \leq T$,
\begin{equation*}\label{limHidreform}
 \lim _{N\to\infty } \PP_{\mu _{N}}\left( \eta_{\cdot}^{N} \in \mathcal D([0,T], {\Omega_{N}}) : \left\vert \dfrac{1}{N}\sum_{x \in \Lambda_{N} }G\left(\tfrac{x}{N} \right)\eta_{x}(tN^{\gamma}) - \int_{0}^1G(u)\rho_{t}^{\kappa}(u)du \right\vert    > \delta \right)= 0,
\end{equation*}
where the time scale is given by $\Theta(N)= N^{\gamma + \theta}$ and  $\rho_{t}^{\kappa}$ is the unique weak solution of:
\begin{itemize}
\item[$\bullet$] \eqref{eq:Dirichlet Equation_infty} with $\hat \kappa=\kappa $, if $\theta<0$;
\item [$\bullet$] (\ref{eq:Dirichlet Equation}) with $\hat \kappa=\kappa $, if $\theta =0$.
\end{itemize}
\end{thm}

Once the hydrodynamic limit is obtained, we would like to know  how  \pat{the  weak solution $\rho_t^{\kappa}$ and the stationary solution $\bar\rho^{\kappa}$ behave as $\kappa$ goes to $0$ or $\infty$ and this is the purpose of  Theorem \ref{convergence_rho^k_to_rho^0} and \ref{Existence_uniqueness_convergence_stationary} stated below. This limiting profile  will give us an idea of what to expect at the hydrodynamics level when we consider our microscopic dynamics in  contact with reservoirs  whose strength is regulated by $\kappa/N^\theta$ and when $\theta\neq 0$ as   in \cite{BGJO}. As mentioned  in the introduction we do not analyze the system in this regime but we conjecture that for small positive values of $\theta>0$ (that corresponds to slow reservoirs) the hydrodynamic limit should be given by the weak solution of (\ref{eq:Dirichlet Equation}) with $\kappa=0$ while for the case $\theta<0$ (that corresponds to fast reservoirs) it should be given by the weak solution of (\ref{eq:Dirichlet Equation_infty}). }

\begin{thm}
\label{convergence_rho^k_to_rho^0}
Let $\rho_{0}:[0,1]\to [0,1]$ be a measurable function. Further, let  $\rho^{\kappa}$ be the weak solution of (\ref{eq:Dirichlet Equation}), with  initial condition $\rho_{0}$ which is independent of $\kappa$ and let $\hat\rho^{\kappa}_{t}:=\rho^{\kappa}_{t/\kappa}$, for all $t\in [0,T]$. Then
\begin{enumerate}[i)]
\item  $\rho^{\kappa}$ converges strongly to  $\rho^{0}$  in $L^{2}(0,T;\mc H^{\gamma/2})$ as $\kappa$ goes to $0$, where $\rho^{0}$ is the weak solution of \eqref{eq:Dirichlet Equation} with $\kappa =0$ and initial condition $\rho_{0}.$
\item {If $\rho_{0} - \bar\rho^{\infty} \in \mc H^{\gamma/2}$ } then $\hat\rho^{\kappa}$ converges strongly to  $\rho^{\infty}$  in $L^{2}(0,T;L^{2}_{V_{1}})$ as $\kappa$ goes to $\infty$, where $\rho^{\infty}$ is the weak solution of \eqref{eq:Dirichlet Equation_infty}. 
\end{enumerate} 
\end{thm}

\begin{rem}
The convergence in Theorem \ref{convergence_rho^k_to_rho^0} is also true in $L^{2}(0,T;L^{2})$. In fact, we will see that a crucial step in the proof of the theorem  is to show that $\rho^{\kappa}$ converges strongly  in $L^{2}(0,T;L^{2})$. Convergence in i) is also true in $L^{2}(0,T; L_{V_{1}}^{2})$ and it is  {a consequence} of the fractional Hardy's inequality (see e.g. \cite{Dy}).
\end{rem}
\begin{thm}
\label{Existence_uniqueness_convergence_stationary}
Let $\bar\rho^{\kappa}$  be the weak solution  of (\ref{eq:Stationary_RFRD}). Then,
\begin{enumerate}[i)]
\item  $\bar\rho^{\kappa}$ converges strongly to  $\bar\rho^{0}$  in $\mc H^{\gamma/2}$ as $\kappa$ goes to $0$, where  $\bar\rho^{0}$ is the weak solution of (\ref{eq:Stationary_RFRD}) with $\kappa =0$.
\item $\bar\rho^{\kappa}$ converges strongly to  $\bar\rho^{\infty}$  in $L^{2}_{V_{1}}$ as $\kappa$ goes to $\infty$, where $\bar\rho^{\infty}$  is given in Remark \ref{rem:Explicit_sol_rho_infty}.
\end{enumerate}  
\end{thm}

\section{Proof of Theorem \ref{limHidreform}: Hydrodynamic limit}\label{sec:proof_hyd_limit}
The proof of this theorem follows the usual approach of convergence in distribution of stochastic processes: we prove tightness  of the sequence  $\lbrace\mathbb{Q}_{N}\rbrace_{N\geq 1}$ and then we prove uniqueness of the limiting point, {which we denote by $\mathbb Q$}. These two results combined give the convergence of  $\lbrace\mathbb{Q}_{N}\rbrace_{N\geq 1}$ to $\mathbb Q$, as $N\rightarrow \infty$. In order to characterize the limiting point $\mathbb Q$,  we prove that all limiting points of  the sequence $\lbrace\mathbb{Q}_{N}\rbrace_{N\geq 1}$ are concentrated on trajectories of  measures that are  absolutely continuous with respect to the Lebesgue measure and whose  density {$\rho_{t}^{\kappa}$} is a weak solution of the  hydrodynamic equation {as given in Definition \ref{Def. Dirichlet Condition}}. From the  uniqueness of the weak solutions of this equation, namely Lemma \ref{existence and uniqueness}, we conclude that $\lbrace\mathbb{Q}_{N}\rbrace_{N\geq 1}$ has a unique limit point $\mathbb{Q}$.  

First, in following subsection we explain how the item iii) in Definition \ref{Def. Dirichlet Condition} appears. In Subsection \ref{subsec:Tightness} we prove that $\{ \mathbb{Q} _{N}\}_{N\geq 1}$ is tight, then in Subsection \ref{subsec:EE} we obtain energy estimates which are crucial to ensure the uniqueness of the limiting point. We conclude this section with the characterization of the limiting point (in Subsection \ref{subsec:Characterization}). 

\subsection{Heuristics for the hydrodynamic equations} \label{sec:CL} 
 In order to make the presentation simple, let us fix a function $G:[0,1]\to\bb R$ which does not depend on time and has compact support included in $(0,1)$.  

By Dynkin's formula (see Lemma A.5.1  in \cite{KL}) we have that
\begin{equation}\label{Dynkin'sFormula}
M_{t}^{N}(G)= \langle \pi_{t}^{N},G\rangle -\langle \pi_{0}^{N},G\rangle-\int_{0}^{t}\Theta(N)L_{N}\langle \pi_{s}^{N},G\rangle ds,
\end{equation}
is a martingale with respect to the natural filtration {$\{\mc F _{t} \}_{t\geq 0}$ where  $\mathcal{F}_{t}:=\sigma(\lbrace\eta(s)\rbrace_{s \leq t})$ for all $t\in [0,T]$.}

Above, for an integrable function $G:[0,1]\rightarrow \RR$, we used the notation {$\langle \pi_{t}^{N},G\rangle$} to represent  the integral of $G$ with respect the measure $ \pi_{t}^{N}$:  
$$\left\langle \pi_{t}^{N},G\right\rangle = \dfrac{1}{N-1}\sum _{x\in \Lambda_{N}}G\left( \tfrac{x}{N}\right)\eta_{x}(t\Theta(N)).$$
In the  previous expression, we are using a measure $\pi_{t}^{N}$ and a function $G$,  therefore, this notation should not  be mistaken  with the one used for the inner product in $L^{2}$. Note that $L_{N}\eta_{x}$ is equal to
$$\sum_{y \in \Lambda_N} p(x-y) [ \eta_{y} -\eta_{x}]+\frac{\kappa}{N^\theta}\sum_{y \leq 0} p(x-y) [ \alpha -\eta_{x}]+\frac{\kappa}{N^\theta}\sum_{y \geq N} p(x-y) [ \beta -\eta_{x}].$$
Therefore, a simple computation shows that 
\begin{equation}
\label{gen_action}
\begin{split}
\Theta(N) L_N ( \langle \pi^N, G \rangle ) =& \cfrac{\Theta(N)}{N-1} \sum_{x\in \Lambda_N}  (\mc L_NG)(\tfrac{x}{N}) \eta_x \\
+& \cfrac{ \kappa \Theta(N)}{N-1} \sum_{x \in \Lambda_N}  G(\tfrac{x}{N}) \left( r_{N}^{-}(\tfrac{x}{N}) (\alpha-\eta_x )+   r_{N}^{+}(\tfrac{x}{N}){(\beta-\eta_x )}\right),
\end{split}
\end{equation}
where, we denote by $\mc L_{N} G$  the continuous function on $[0, 1]$ which is defined as the linear interpolation of the function
\begin{equation}
\label{eq:mcln}
 ({\mc L}_N G) (\tfrac{x}{N}) = \sum_{y \in \Lambda_N} p(y-x) \left[ G(\tfrac{y}{N}) -G(\tfrac{x}{N})\right],
\end{equation}
for all  $x\in \Lambda_N$ with $(\mc L_{N} G)(0) = (\mc L _{N}G)(1) = 0$.
We also define the functions $r_{N}^{\pm}: [0,1]\to \RR$ as the linear interpolation of the function 
\begin{equation}
r_N^- (\tfrac{x}{N})= \sum_{y \ge x} p(y), \quad r_N^+ (\tfrac{x}{N})= \sum_{y \le x-N} p(y),
\end{equation}   
for all $x \in \Lambda_{N}$ with $r_{N}^{\pm}(0) = r_{N}^{\pm}(\tfrac{1}{N})$ and $r_{N}^{\pm}(1) = r_{N}^{\pm}(\tfrac{N-1}{N})$.
By Lemma 3.3 in  \cite{BJ} we have that 
\begin{equation}\label{F_convergencer}
\begin{split}
&\lim _{N\to \infty} N^{\gamma} (r_N^-)(u) =r^-(u),\,\,\,\lim _{N\to \infty} N^{\gamma} (r_N^+)(u) =r^+(u)
\end{split} 
\end{equation}
uniformly in  $[a,1-a]$ for $a\in (0,1)$ and we also can deduce from that {lemma}  that 
\begin{equation}\label{F_convergenceLL}
\lim _{N\to \infty}N^{\gamma}({\mc L}_N G)(u)  = (\LL G)(u)
\end{equation}
uniformly in  $[a,1-a]$, for all functions $G$ with compact support included in $[a,1-a]$. 

Now, we are going to analyse all the terms in \eqref{gen_action} for  $\theta\leq 0$. Thus, we will be able to see how the different boundary conditions appear on the hydrodynamic equations given in Subsection \ref{subsec:hyd_eq} from the underlying particle system.

\subsubsection{\underline{The case $\theta <0$}} 
In this regime we take $\Theta(N)=N^{\gamma+\theta}$ and a function $G\in C_{c}^{\infty}(0,1)$. By using  \eqref{F_convergenceLL} we have that the first term on the right hand side of \eqref{gen_action}  vanishes since $\theta<0$. Now, the second term on the right hand side in (\ref{gen_action}) is equal to $\kappa \langle \alpha -\pi_{t}^{N},  G r^-_{N} \rangle +\kappa \langle \beta -\pi_{t}^{N},  G r^+_{N} \rangle$. By  \eqref{F_convergencer} the previous expression converges, as $N$ goes to $\infty$, to 
\begin{equation*}
\begin{split}
 & \kappa\int_{0}^{1}(\alpha -\rho_{t}^{\kappa}(u))G(u)r^{-}(u)du +\kappa\int_{0}^{1}(\beta-\rho_{t}^{\kappa}(u))G(u)r^{+}(u)du\\
&=-  \kappa\int_{0}^{1}\rho_{t}^{\kappa}(u)G(u)V_{1}(u)du+\kappa\int_{0}^{1}G(u)V_{0}(u)du.\\
\end{split}
\end{equation*}

\subsubsection{\underline{The case $\theta =0$}}
In this regime we take $\Theta(N)=N^{\gamma+\theta}$ and a function $G\in C_{c}^{\infty}(0,1)$. The first term  on the right hand side in (\ref{gen_action}) can be replaced, thanks to (\ref{F_convergenceLL}) by 
$$\langle \pi_{t}^{N}, \LL G \rangle \to \int_{0}^{1}(\LL G)(u) \rho_{t}^{\kappa}(u)du,$$
as $N$ goes to $\infty$. Similarly, the second term on the right hand side of (\ref{gen_action}) is equal to $\kappa \langle \alpha -\pi_{t}^{N},  G r^- \rangle +\kappa \langle \beta -\pi_{t}^{N},  G r^+ \rangle$ which converges, as $N$ goes to $\infty$, to 
\begin{equation*}
\begin{split}
 & \kappa\int_{0}^{1}(\alpha -\rho_{t}^{\kappa}(u))G(u)r^{-}(u)du +\kappa\int_{0}^{1}(\beta-\rho_{t}^{\kappa}(u))G(u)r^{+}(u)du\\
&=-  \kappa\int_{0}^{1}\rho_{t}^{\kappa}(u)G(u)V_{1}(u)du+\kappa\int_{0}^{1}G(u)V_{0}(u)du.\\
\end{split}
\end{equation*}

This intuitive argument is rigorously proved in Subsection \ref{subsec:Characterization}.

\subsection{Tightness} 
\label{subsec:Tightness}
In this subsection we prove that the sequence $\lbrace  \mathbb {Q}_{N} \rbrace_{N \geq 1} $ is tight. We use the usual approach (see, for example, Proposition 4.1.6 in \cite{KL}), which says that is enough to show that, for all $\ve >0$ 
\begin{equation}
\label{T1}
\displaystyle \lim _{\delta \rightarrow 0} \limsup_{N\rightarrow\infty} \sup_{\tau  \in \mathcal{T}_{T},\bar\tau \leq \delta} {\mathbb{P}}_{\mu _{N}}\Big[\eta_{\cdot}^{N}\in {\mathcal D} ( [0,T], \Omega_{N}) :\left\vert \langle\pi^{N}_{\tau+ \bar\tau},G\rangle-\langle\pi^{N}_{\tau},G\rangle\right\vert > \ve \Big]  =0, 
\end{equation}
for any function $G$ belonging to $C([0,1])$ . Above $\mathcal{T}_{T}$ is the set of stopping times bounded by $T$ and we implicitly assume that all the stopping times are bounded by $T$, thus, $\tau+ \bar\tau$ should be read as $ (\tau+ \bar\tau) \wedge T$. Indeed, we prove below that (\ref{T1}) is true for any function $G$ in $C_{c}^{2}((0,1))$, by using an $L^1$ approximation procedure(a similar argument as done in \cite{BGJO}), we can extend this class of functions to functions $G\in C([0,1])$.

\begin{prop}\label{Tightness}
The sequence of measures $\lbrace\mathbb{Q}_{N}\rbrace_{N\geq 1}$ is tight with respect to the Skorohod topology of $\mathcal D([0, T],{\mathcal{M^{+}}})$.
\begin{proof}
 Note that, we are going to prove (\ref{T1}) for functions $G$ in  $C_{c}^{2}((0,1))$. Recall from (\ref{Dynkin'sFormula})  that $M_{t}^{N}(G)$ is a martingale with respect to the natural filtration  $\{\mathcal{F}_{t}\}_{t\geq 0}$. In order to prove (\ref{T1}) it is enough to show that
\begin{eqnarray} 
\label{TC1}
\displaystyle \lim _{\delta\rightarrow 0} \limsup_{N\rightarrow\infty} \sup_{\tau  \in \mathcal{T}_{T},\bar\tau \leq \delta}\mathcal{\mathbb{E}}_{\mu _{N}}\left[ \left\vert \int_{\tau}^{\tau+ \bar\tau}N^{\gamma} L_{N}\langle \pi_{s}^{N},G\rangle ds \right\vert\right] = 0
\end{eqnarray}
and
 \begin{equation} 
 \label{TC2}
\displaystyle \lim _{\delta\rightarrow 0} \limsup_{N\rightarrow\infty} \sup_{\tau  \in \mathcal{T}_{T},\bar\tau \leq \delta}\mathcal{\mathbb{E}}_{\mu _{N}}\left[\left( M_{\tau}^{N}(G)- M_{\tau+ \bar\tau}^{N}(G) \right)^{2}  \right]=0.
\end{equation}

\vspace{0.5cm}
\noindent 
By using (\ref{F_convergencer}), (\ref{F_convergenceLL}) and the fact that $G\in C_{c}^{2}((0,1))$ we can bound the expression in (\ref{gen_action}) by a constant. By using the fact that $|\eta_x^N (s)| \le 1$  and 
\begin{equation}\label{Tx}
\displaystyle\sum_{x\geq 1}\left( r_{N}^{-}(\tfrac{x}{N})+ r_{N}^{+}(\tfrac{x}{N})\right)<\infty
\end{equation}
(since $\gamma >1$), we can bound from above the second term  at the right hand side in (\ref{gen_action}) by a constant times $\Theta(N) N^{-1-\theta}$. Considering the different values of $\theta$ we see that such term is bounded from above by a constant. Then we have that 
\begin{equation}
\vert \Theta(N) L_N ( \langle \pi^N_{s}, G \rangle )\vert \ls  1
\end{equation}
for any $s\le T$, which trivially implies \eqref{TC1}.

In order to prove \eqref{TC2}, 
 by Dynkin's formula (see Appendix 1  in \cite{KL}) we know that  $$\displaystyle\left( M^{N}_{t}(G)\right)^{2}-\int^{t}_{0} \Theta(N)\left[ L_{N} \langle\pi^{N}_{s},G \rangle^{2}- 2\langle\pi^{N}_{s},G \rangle L_{N} \langle\pi^{N}_{s},G \rangle\right]ds,$$
is a martingale with respect to the natural filtration  $\{\mathcal{F}_{t}\}_{t\ge 0}$. By Lemma \ref{lem:compA} we get that the term inside the time integral in the previous expression is equal to 
\begin{equation}\label{T5}
\begin{split}
&\dfrac{\Theta(N)}{(N-1)^{2}} \sum_{x<y\in\Lambda_{N}}\left( G\left(\tfrac{x}{N} \right) -G\left(\tfrac{y}{N} \right)\right)^{2}p(x-y)(\eta_{y}(s\Theta(N))-\eta_{x}(s\Theta(N)))^{2}\\
+&\dfrac{\kappa \Theta(N)}{(N-1)^{2}}\sum_{ x\in\Lambda_{N}} \left(G\left(\tfrac{x}{N}\right)\right)^{2}(1-2\eta_{x}(s\Theta(N))) \left(r^{-}_{N}(\tfrac{x}{N})(\alpha-\eta_{x}(s\Theta(N))) + r^{+}_{N}(\tfrac{x}{N})(\beta-\eta_{x}(s\Theta(N)))\right).\\
\end{split}
\end{equation}
Since  the first derivative of $G$ is bounded it is easy to see that the absolute value of (\ref{T5}) is bounded from above by a constant times
\begin{equation}
\label{T6}
\begin{split}
&\dfrac{\Theta(N)}{(N-1)^{4}} \sum_{x,y\in\Lambda_{N}} (x-y)^{2}p(x-y)
+\dfrac{\kappa \Theta(N)}{(N-1)^{2}}\sum_{ x\in\Lambda_{N}} \left(G\left(\tfrac{x}{N}\right)\right)^{2} \left( r_{N}^{-}(\tfrac{x}{N}) +r_{N}^{+}(\tfrac{x}{N})  \right).\\
\end{split}
\end{equation}
Note that  $ (x-y)^{2}p(x-y) \ls 1$ because $\gamma > 1$, {so that}  
$$\dfrac{\Theta(N)}{(N-1)^{4}} \sum_{x,y\in\Lambda_{N}} (x-y)^{2}p(x-y) \ls \Theta(N) N^{ -2}= \mc{O} (N^{\gamma-2}).$$ 
By \eqref{Tx}, the remaining terms in (\ref{T6}) are $\mc{O}( \Theta(N)N^{-\theta-2})$  so that \eqref{T6} is $\mc{O} (N^{\gamma-2})$.

Thus, since $\tau$ is a stopping time and $\gamma< 2$ we have that  
\begin{eqnarray*} 
& &\displaystyle \lim _{\delta\rightarrow 0} \limsup_{N\rightarrow\infty} \sup_{\tau  \in \mathcal{T}_{T},\bar\tau \leq \delta}\mathcal{\mathbb{E}}_{\mu ^{N}}\left[\left( M_{\tau}^{N,G}- M_{\tau+ \bar\tau}^{N,G} \right)^{2}  \right]\\
=& &\displaystyle \lim _{\delta\rightarrow 0} \limsup_{N\rightarrow\infty} \sup_{\tau  \in \mathcal{T}_{T},\bar\tau \leq \delta}\mathcal{\mathbb{E}}_{\mu ^{N}}\left[ \int^{\tau+\bar\tau}_{\tau}\Theta(N)\left[ L_{N} \langle\pi^{N}_{s},G \rangle^{2}- 2\langle\pi^{N}_{s},G \rangle L_{N} \langle\pi^{N}_{s},G \rangle\right]ds \right]\\
=&&0.
\end{eqnarray*}

Therefore, we have proved (\ref{T1}) for functions $G$ in $C_{c}^{2}((0,1))$ and as we have said in the beginning of the subsection this is enough to conclude tightness.
\end{proof}
\end{prop}
\subsection{Energy Estimate}
\label{subsec:EE}
We prove in this subsection that any limit point $\mathbb{Q}$ of the sequence $\lbrace\mathbb{Q}_{N}\rbrace_{N\geq 1}$ is concentrated on trajectories $\pi_{t}^{\kappa}(u)du$ with finite energy, i.e., $\pi^{\kappa}$ belongs to $L^{2}(0,T;\mc H^{\gamma/2})$. Moreover, we prove that $\pi_{t}^{\kappa}$ satisfies item ii) in Definition \ref{Def. Dirichlet Condition}. The latter is the content of Theorem  \ref{Energy_Thm1}  stated below. Fix a limit point $\mathbb{Q}$ of the sequence $\lbrace\mathbb{Q}_{N}\rbrace_{N\geq 1}$ and  assume, without of loss of generality, that the sequence $\mathbb{Q}_{N}$ converges to $\mathbb{Q}$ as $N$ goes to $\infty$.
\begin{thm}\label{Energy_Thm1}
The probability measure $\mathbb{Q}$ is concentrated on trajectories of measures of the form $\pi_{t}^{\kappa}(u)du$, such that for any interval $I \subset [0,T]$  the density $\pi^{\kappa}$ satisfies 
\begin{enumerate}[i)] 
\item  $\int _{I} \Vert \pi^{\kappa}_{t} \Vert _{\gamma/2}^{2}dt \ls \vert I\vert (\kappa + 1)$, if $\theta=0$.  
\item $\displaystyle\int_I\int_{0}^{1}\left\lbrace \dfrac{(\alpha-\pi_{t}^{\kappa}(u))^{2}}{u^{\gamma}} + \dfrac{(\beta-\pi_{t}^{\kappa}(u))^{2}}{(1-u)^{\gamma}}\right\rbrace du\;dt \ls \vert I\vert\dfrac{\kappa+1}{\kappa}$, if $\theta\leq 0$.
\end{enumerate}
\end{thm}
\begin{rem}\label{Energy_Rm1}
It follows from item i) of the previous and from Theorem 8.2 of  \cite{dPV} that $\pi_{t}^{\kappa}$ is, $\PP$ almost surely,  $\tfrac{\gamma-1}{2}$-H\"older for all $t\in I$.

By taking $I= [0,T]$ in item i) of Theorem \ref{Energy_Thm1} we can see that $\pi^{\kappa} \in L^{2}(0,T;\mc H ^{\gamma/2})$. Moreover, from item ii) of Theorem \ref{Energy_Thm1}, we claim that
\begin{equation} \label{EE_3}
\int_{I}\Vert \pi^{\kappa}_{t} - \bar\rho^{\infty}\Vert_{V_{1}}^{2}dt \ls |I|\dfrac{\kappa +1}{\kappa}
\end{equation}
where $\bar\rho^{\infty}$ is given in Remark \ref{rem:Explicit_sol_rho_infty}. Note that
\begin{equation}\label{EE_R1}
\begin{split}
\int_{I}\Vert \pi^{\kappa}_{t} - \bar\rho^{\infty}\Vert_{V_{1}}^{2}dt &= c_{\gamma}\gamma^{-1}\int_{I}\int_{0}^{1}\left\lbrace \dfrac{( \pi^{\kappa}_{t}(u) - \bar\rho^{\infty}(u))^{2}}{u^{\gamma}}+\dfrac{( \pi^{\kappa}_{t}(u) - \bar\rho^{\infty}(u))^{2}}{(1-u)^{\gamma}}\right\rbrace dudt.
\end{split}
\end{equation} 
By summing and subtracting $\alpha$ inside the first square in the expression on the right hand side in (\ref{EE_R1}), $\beta$ in the second one and using the fact that $(a+b)^{2}\leq 2(a^{2}+b^{2})$ we get that  (\ref{EE_R1}) is bounded from above by 
\begin{equation}\label{EE_R2}
\begin{split}
& 2 c_{\gamma}\gamma^{-1}\int_{I}\int_{0}^{1}\left\lbrace \dfrac{( \pi^{\kappa}_{t}(u) - \alpha)^{2}}{u^{\gamma}}+\dfrac{( \pi^{\kappa}_{t}(u) -\beta)^{2}}{(1-u)^{\gamma}}\right\rbrace dudt\\
+& \, 2c_{\gamma}\gamma^{-1}\int_{I}\int_{0}^{1}\left\lbrace \dfrac{( \alpha- \bar\rho^{\infty}(u))^{2}}{u^{\gamma}}+\dfrac{( \beta - \bar\rho^{\infty}(u))^{2}}{(1-u)^{\gamma}}\right\rbrace dudt.\\
\end{split}
\end{equation}
Now, by using item ii) of Theorem \ref{Energy_Thm1} we have that the first term in the previous expression is bounded by constant times $|I|\dfrac{\kappa +1}{\kappa}$. Finally, using the definition of $\bar\rho^{\infty}$ (see Remark \ref{rem:Explicit_sol_rho_infty}) the second term in \eqref{EE_R2} is equal to 
\begin{equation*}
 2c_{\gamma}\gamma^{-1}(\beta -\alpha)^{2}|I|\int_{0}^{1}(u^{\gamma}+(1-u)^{\gamma})^{-1} du\ls 1.
\end{equation*}
\end{rem}

Before we prove Theorem \ref{Energy_Thm1}, we establish some estimates on the Dirichlet form which are needed in due course.
\subsubsection{Estimates on the  Dirichlet form}
Let 
$h : [0, 1] \rightarrow [0, 1]$ be a function such that  $\alpha \leq h(u)\leq\beta$, for all $u\in [0,1]$, and assume that $h(0) =\alpha$ and $h(1)=\beta$. Let $\nu_{h}^{N}$ be the inhomogeneous Bernoulli product measure on 
$\Omega_{N}$ with marginals given by
$$\nu_{h}^{N}\lbrace\eta: \eta_{x} = 1 \rbrace = h\left( \tfrac{x}{N}\right).$$

We denote  by $H_{N}(\mu\vert \nu_{h}^{N})$ the relative entropy of a probability measure $\mu$ on 
$\Omega_{N}$ with respect to the probability measure $\nu_{h}^{N} $. It is easy to prove the existence of a constant $C_0$, such that
\begin{equation}\label{H}
H_{N}(\mu_{N}\vert \nu_{h}^{N})\leq C_0 N.
\end{equation}
(see for example \cite{BGJO}). \pat{We remark here that the restriction $\alpha\neq 0$ and $\beta\neq 1$ comes from last estimate since the constant $C_0$ given above is given by $C_0=-\log(\alpha \wedge(1-\beta))$.} On the other hand, for a probability measure $\mu$ on $\Omega_N$ and a density function $f:\Omega_N \to [0,\infty)$ with respect to $\mu$ we introduce  
\begin{eqnarray} 
\label{left_rig_form}
D_{N}^{0}(\sqrt{f},\mu):=\cfrac{1}{2}\sum_{x,y\in\Lambda_N}p(y-x)\, I_{x,y}(\sqrt{f},\mu), 
\end{eqnarray}
\begin{eqnarray}
\label{left_dir_form}
D_{N}^{\ell}(\sqrt{f},\mu):=\sum_{x\in\Lambda_N}\sum_{y\leq 0}p(y-x)\, I^\alpha_{x}(\sqrt{f},\mu)=\sum_{x\in\Lambda_N}r_N^-(\tfrac{x}{N})I^\alpha_{x}\, (\sqrt{f},\mu) 
\end{eqnarray} 
and $D_{N}^{r}(\sqrt{f},\mu)$ is the same as $D_{N}^{\ell}(\sqrt{f},\mu)$ but in $I^\alpha_{x}(\sqrt{f},\mu)$ the parameter $\alpha$ is replaced by $\beta$ and $r_N^-$ is replaced by $r_N^+$.  Above, we used the following notation
\begin{eqnarray*}
I_{x,y}(\sqrt f,\mu)&:=& \int \left(\sqrt {f(\sigma^{x,y}\eta)}-\sqrt {f(\eta)}\right)^{2} d\mu,\\
I_{x}^{\alpha}(\sqrt f,\mu)&:=& \int  c_{x}(\eta;\alpha)\left(\sqrt {f(\sigma^{x}\eta)}-\sqrt {f(\eta)}\right)^{2} d\mu
\end{eqnarray*}
and $I_{x}^{\beta}$ is the same as $I_{x}^{\alpha}$ when the parameter $\alpha$ is replaced by $\beta$.

Our goal is to express, for the measure $\nu_{h}^{N}$,  a relation between the Dirichlet form defined by $\langle L_N\sqrt{f},\sqrt{f} \rangle_{\nu_{h}^{N}}$ and the quantity
\begin{eqnarray*}
D_{N}(\sqrt{f},\nu_{h}^{N} )&:=& (D_{N}^{0}+\kappa N^{-\theta} D_{N}^{\ell}+\kappa N^{-\theta}D_{N}^{r})(\sqrt{f},\nu_{h}^{N}). \end{eqnarray*}
 
More precisely, we have the following result.
\begin{lem}\label{bound_Dir}
For any positive constant $B$ and any density function $f$ with respect to $\nu_{h}^N$, there exists a constant $C>0$ (independent of $f$ and $N$) such that 
\begin{equation}
\label{dir_est}
\begin{split}
\frac{\Theta(N)}{NB}\langle L_{N}\sqrt{f},\sqrt{f} \rangle_{\nu_{h}^N} &\leq -\dfrac{\Theta(N)}{4NB}D_{N}(\sqrt{f},\nu_{h}^N) + \dfrac{C\Theta(N)}{N	B}\sum_{x,y\in\Lambda_N}p(y-x)\Big(h(\tfrac xN)-h(\tfrac yN)\Big)^2\\
& + \dfrac{C\kappa\Theta(N)}{N^{\theta+1}B} \sum_{x\in\Lambda_N}\left\lbrace \Big(h(\tfrac xN)-\alpha\Big)^2r^{-}_{N}(\tfrac{x}{N}) + \Big(h(\tfrac xN)-\beta\Big)^2r^{+}_{N}(\tfrac{x}{N}) \right\rbrace.
\end{split}
\end{equation}
\end{lem}
The proof of this statement is similar to the one in Section 5 of \cite{BGJO} and thus it is omitted. Moreover, note that as a consequence of the previous lemma, for a function $h$ such that $\alpha \leq h(u)\leq \beta$ and $h$ Lipschitz we have that
\begin{equation}
\label{dir_est_lip}
\begin{split}
\frac{\Theta(N)}{NB}\langle L_{N}\sqrt{f},\sqrt{f} \rangle_{\nu_{h}^N} &\leq -\dfrac{\Theta(N)}{4NB}D_{N}(\sqrt{f},\nu_{h}^N) + \Theta (N)N^{-\gamma}\frac{C(\kappa N^{-\theta}+1)}{B}.
\end{split}
\end{equation}

\begin{lem}
 \label{bound}
For any density $f$ with respect to $\nu_{h}^N$,  any $x\in \Lambda_{N}$ and any positive constant $A_x$, we have that
$$\left\vert \left\langle \eta_x-\alpha,f\right\rangle_{\nu_{h}^{N}} \right\vert \; \ls \; \dfrac{1}{4A_{x}}  I_{x}^{\alpha}(\sqrt{f},\nu_{h}^{N})+ A_{x}+\left\vert h(\tfrac{x}{N})- \alpha\right\vert.$$
 The same result holds if $\alpha$ is replaced by $\beta$.
\end{lem}
The proof of Lemma \ref{bound} is omitted since is similar to  the one of Lemma 5.5 in \cite{BGJO}. Note that in the case $\alpha \le h \le \beta$ and  Lipschitz  we get 
$$\left\vert \left\langle \eta_x-\alpha,f\right\rangle_{\nu_{h}^{N}} \right\vert \; \ls\; \dfrac{1}{4A_{x}}  I_{x}^{\alpha}(\sqrt{f},\nu_{h}^{N})+ A_{x}+\dfrac{x}{N}.$$

\subsubsection{Proof of of Theorem \ref{Energy_Thm1}}
\textbf{First step:} $\pi^{\kappa} \in L^{2}(0,T;\mc H ^{\gamma/2})$ $\mathbb{Q}$ almost surely. Recall that in this case ($\theta=0$) the system is speeded up in the sub-diffusive time scale $\Theta(N)=N^{\gamma}$. Let $\ve>0$ be a small real number. Let $F \in C_c^{0,\infty} (I\times[0,1]^{2})$, where the I is a  subinterval of $[0,T]$. By the  entropy and Jensen's inequality and Feynman-Kac's formula (see Lemma A.7.2 in \cite{KL}),  we have that  
\begin{equation}\label{eq:varfor}
\begin{split}
&{\mathbb E}_{\mu_N} \Big[\int_I  \;  N^{\gamma -1} \sum_{\substack{x,y\in \Lambda_N\\ |x-y|\geq \varepsilon N} } F_{t}( \tfrac{x}{N}, \tfrac{y}{N})p(y-x) (\eta_y(t\Theta(N))-\eta_x(t\Theta(N)))\Big]dt \\
&\leq C_{0} + \int_{I}\sup_{f} \Big\{ N^{\gamma-1}  \sum_{\substack{x,y\in \Lambda_N\\ |x-y|\geq \varepsilon N} } F_{t}( \tfrac{x}{N}, \tfrac{y}{N})p(y-x) \int(\eta_y-\eta_x)f(\eta)d\nu^{N}_{h}+  N^{\gamma-1} \left\langle  L_N   {\sqrt f} , {\sqrt f} \right\rangle_{\nu^{N}_{h}}  \Big\}dt
\end{split}
\end{equation}

 where the supremum is taken over all  densities $f$ on $\Omega_N$ with respect to $\nu_{h}^{N}$. 
Note that, by a change of variables, we have that
\begin{equation}
\label{eq:pat67}
\begin{split}
&N^{\gamma-1} \sum_{\substack{x,y \in \Lambda_N\\ |x-y| \ge \ve N}}  F_{t}( \tfrac{x}{N}, \tfrac{y}{N} ) p(y-x) \int(\eta_y -\eta_x) f(\eta) d\nu^{N}_{h} \\
=& N^{\gamma-1} \sum_{\substack{x,y \in \Lambda_N\\ |x-y| \ge \ve N}}  F_{t}^a( \tfrac{x}{N}, \tfrac{y}{N} ) p(y-x) \int (\eta_y -\eta_x) f(\eta) d\nu^{N}_{h}\\
=&N^{\gamma-1}\sum_{\substack{x,y \in \Lambda_N\\ |x-y| \ge \ve N}}  F_{t}^a( \tfrac{x}{N},\tfrac{y}{N} )  p(y-x) \int \eta_y \left( f(\eta) - f(\sigma^{x,y}\eta)\right) d\nu^{N}_{h}\\
+&N^{\gamma-1}\sum_{\substack{x,y \in \Lambda_N\\ |x-y| \ge \ve N}} F_{t}^a( \tfrac{x}{N}, \tfrac{y}{N} ) p(y-x) \int \eta_x f(\eta)\left(\theta^{x,y}(\eta) -1\right)d\nu^{N}_{h}
\end{split}
\end{equation}
where $\theta^{x,y}(\eta)=\tfrac{d\nu_{h}^{N}(\sigma^{x,y}\eta)}{d\nu_{h}^{N}(\eta)}$ and  $F^a$ is the antisymmetric part of $F$, i.e. for all $t\in I$ and $(u,v) \in [0,1]^2$ 
$$F_{t}^a (u,v) =\cfrac{1}{2} \Big[ F_{t}(u,v) -F_{t}(v,u) \Big].$$
Observe that $F_{t}^{a} (u,u)=0$. 
By  Young's inequality, the fact that $f$ is a density and $|\eta_y| \le 1$, we have that, for any $A>0$,  the third term in \eqref{eq:pat67} is bounded from above by a constant times
\begin{equation*}
\begin{split}
& N^{\gamma-1}A \sum_{\substack{x,y \in \Lambda_N\\ |x-y| \ge \ve N}} \Big( F_{t}^a \Big( \tfrac{x}{N}, \tfrac{y}{N} \Big)\Big)^2 p(y-x) +\frac{N^{\gamma-1}}{A} \sum_{\substack{x,y \in \Lambda_N\\ |x-y| \ge \ve N}} p(y-x)  I_{x,y}(\sqrt {f}, \nu^N_{h})\\
 \leq&\frac {c_\gamma A}{N^2}   \sum_{\substack{x,y \in \Lambda_N\\ |x-y| \ge \ve N}}  \cfrac{\Big( F_{t}^a \Big( \tfrac{x}{N}, \tfrac{y}{N} \Big)\Big)^2}{ | \tfrac{x}{N} -\tfrac{y}{N}|^{1+\gamma}} \; + \; \dfrac{2N^{\gamma-1}}{A}D^{0}_{N}(\sqrt {f}, \nu^N_{h}). 
 \end{split}
\end{equation*}
Since $h$ is Lipschitz we have that $\sup_{\eta \in \Omega_N} \, \left|\theta^{x,y}(\eta) -1\right| =  \mc{O} \left( \tfrac{|x -y|}{N}\right)$.
By Young's inequality and the fact that $f$ is a density, for any $A^{'}>0$, the last term in \eqref{eq:pat67} is bounded from above by 
\begin{equation*}
\begin{split}
& \frac{N^{\gamma-1}}{A^{'}}\sum_{\substack{x,y \in \Lambda_N\\ |x-y| \ge \ve N}}   \Big( F_{t}^a \Big( \tfrac{x}{N}, \tfrac{y}{N} \Big)\Big)^2 p(y-x)  \; + \; A^{'}N^{\gamma-1}  \sum_{\substack{x,y \in \Lambda_N\\ |x-y| \ge \ve N}}  p(y-x) \left(\tfrac{|x-y|}{N}\right)^{2}\\
&=\frac {c_\gamma}{A^{'}N^2}   \sum_{\substack{x,y \in \Lambda_N\\ |x-y| \ge \ve N}}  \cfrac{\Big( F_{t}^a \Big( \tfrac{x}{N}, \tfrac{y}{N} \Big)\Big)^2}{ | \tfrac{x}{N} -\tfrac{y}{N}|^{1+\gamma}}   \; + \;    \cfrac{A^{'} c_{\gamma}}{N^{2}}  \sum_{\substack{x,y \in \Lambda_N\\ |x-y| \ge \ve N}} \frac{1}{|\tfrac{x}{N}-\tfrac{y}{N}|^{\gamma -1}}.
 \end{split}
\end{equation*}
Recall \eqref{dir_est_lip}, so that by choosing $A=8$ and $B=1$ and using the two results above we have just proved that \eqref{eq:varfor} is bounded from above by $C_0$ plus 
\begin{equation*}
  \cfrac{c_{\gamma}(8+ \tfrac{1}{A^{'}})}{N^2} \sum_{x\ne y \in \Lambda_N}  \cfrac{\left[ F_{t}^a( \tfrac{x}{N},\tfrac{y}{N} ) \right]^2 }{| \tfrac{x}{N} -\tfrac{y}{N}|^{1+\gamma}}  + C(\kappa +1)+c_{\gamma} A^{'} A^{''},
\end{equation*}
where 
$$A^{''}:=\sup_{\ve >0} \sup_{N \ge 1} \cfrac{1}{N^2}  \sum_{\substack{x,y \in \Lambda_N\\ |x-y| \ge \ve N}} \frac{1}{|\tfrac{x}{N}-\tfrac{y}{N}|^{\gamma -1}} <\infty$$
since $\gamma<2$. Therefore, we have proved that there exist constants $A^{'''}$ and  $B^{'}$ (independent of $\ve > 0$, $N\geq 1$, and $F\in C_{c}^{\infty}(I\times[0,1]^{2})$) such that
\begin{equation}\label{eq:imp}
\begin{split}
&{\mathbb E}_{\mu_N} \bigg[ \int_I   N^{\gamma -1}  \sum_{\substack{x,y\in \Lambda_N\\ |x-y|\geq \varepsilon N} } F_{t}( \tfrac{x}{N}, \tfrac{y}{N})p(y-x) (\eta_y^{N}(t)-\eta_x^{N}(t))\,dt\bigg]\\
&={\mathbb E}_{\mu_N} \left[ \int_I -2c_{\gamma}\langle\pi^N_t,g^N_t\rangle\,dt\right] \\
&\le \int _I\cfrac{A^{'''} }{N^2} \sum_{\substack{x,y \in \Lambda_N\\ |x-y| \ge \ve N}}  \cfrac{c_\gamma \left( F_{t}^a ( \tfrac{x}{N}, \tfrac{y}{N} ) \right)^2}{| \tfrac{x}{N} -\tfrac{y}{N}|^{1+\gamma}}  dt+ B^{'}| I|(\kappa  +1).
\end{split}
\end{equation}
Above the function $g^{N}$ is defined on $ I\times[0,1]$ by
$$g_t^{N} (u) = \cfrac{1}{N} \sum_{y \in \Lambda_N} \, \textbf{1}_{ \big|\tfrac{y}{N}- u \big| \ge \ve }\cfrac{F_{t}^a \big(u, \tfrac{y}{N}\big)}{\vert u -\tfrac{y}{N}\vert^{1+\gamma}}$$
and it is a discretization of the smooth function  $g$ defined on $(t,u) \in I\times[0,1]$  by
$$ \quad g_{t}(u) = \int_0^1\textbf{1}_{\{ |v-u| \ge \ve\}} \cfrac{F_{t}^a (u,v)}{|u-v|^{1+\gamma}} \, dv.$$
Let $Q_{\ve}=\{(u,v)\in [0,1]^2\; ; \; |u-v| \ge \ve\}$. Observe first that for symmetry reasons we have that for any integrable function $\pi$, 
$$\int_0^1 \pi (u) g_{t}(u) du =   \iint_{Q_\ve}\cfrac{(\pi (v) -\pi (u)) F_{t}^{a} (u,v) }{|u-v|^{1+\gamma}}\; du dv.$$

 By taking the limit as $N \to \infty$ in \eqref{eq:imp}, we conclude that there exist constants $C>0$ independent of $F \in C_c^{0,\infty} (I\times[0,1]^2)$ and $\ve>0$ such that
\begin{equation*}
\begin{split}
&{\bb E}_{\mathbb{Q}} \left[  \int _{I}\iint_{Q_\ve}\cfrac{(\pi_{t}^{\kappa}  (v) -\pi_{t} ^{\kappa} (u)) F_{t}^{a} (u,v) }{|u-v|^{1+\gamma}}\; - C\cfrac{ \big( F^{a}_{t} (u,v) \big)^2 }{|u-v|^{1+\gamma}} \; du dv dt  \right] \ls |I|(\kappa +1).
\end{split}
\end{equation*}
From Lemma 7.5 in \cite{KLO2} we can insert the  supremum over $F$ inside the expectation above, so that
\begin{equation*}
\begin{split}
&{\bb E}_{\mathbb{Q}}  \left[ \sup_F \left\{ \int _{I}  \iint_{Q_{\ve} }\cfrac{(\pi_{t}^{\kappa}  (v) -\pi_{t} ^{\kappa} (u)) F_{t}^{a} (u,v) }{|u-v|^{1+\gamma}}\;
 -  C \cfrac{ \big( F_{t}^{a} (u,v) \big)^2 }{|u-v|^{1+\gamma}} \; du dv dt    \right\} \right]\ls  |I|(\kappa +1).
\end{split}
\end{equation*}
Since the function $(u,v) \in [0,1]^2 \to \pi(v) -\pi (u)$ is antisymmetric we may replace $F^a$ by $F$ in the previous variational formula, i.e.
\begin{equation}
\label{eq:A3}
\begin{split}
&{\bb E}_{\mathbb{Q}} \left[ \sup_F \left\{ \int _{I}  \iint_{Q_{\ve} }\cfrac{(\pi_{t}^\kappa (v) -\pi_{t}^{\kappa} (u)) F _{t}(u,v) }{|u-v|^{1+\gamma}}\;  - C \cfrac{ \big( F _{t}(u,v) \big)^2 }{|u-v|^{1+\gamma}} \; du dv dt    \right\} \right] \ls|I|(\kappa+1).
\end{split}
\end{equation}

Consider the Hilbert space ${\mathbb L}^2 ([0,1]^2, d\mu_{\ve})$ where $\mu_{\ve}$ is the measure whose density with respect to Lebesgue measure is given by
$ (u,v) \in [0,1]^2 \to {\bb 1}_{|u-v| \ge \ve} \, |u-v|^{-(1+\gamma)}.$
By taking 
$$\Pi^{\kappa}: (t;u,v) \in I \times[0,1]^2 \to \pi_{t}^{\kappa}(v) -\pi_{t} ^{\kappa}(u),$$ 
the previous formula implies that 
\begin{equation}
\label{1S}
\EE_{\mathbb{Q}}  \left[ \int _{I}\iint_{[0,1]^2} \left(\Pi^{\kappa} _{t}(u,v)\right)^2 \, d\mu_{\ve} (u,v)dt \right] \ls |I|(\kappa+1).
\end{equation}
Letting $\ve \to 0$, by the monotone convergence theorem, we conclude that
\begin{equation*}
 \int_{I}\iint_{[0,1]^2}\cfrac{(\pi_{t}^{\kappa} (v) -\pi_{t}^{\kappa} (u))^2}{|u-v|^{1+\gamma}}\; du dvdt <\infty
  \end{equation*}
 $\mathbb{Q}$ almost surely. 
 
 \vspace{.5cm}
 
\noindent\textbf{Second step:} $\displaystyle\int_I\int_{0}^{1}\left\lbrace \dfrac{(\alpha-\pi_{t}^{\kappa}(u))^{2}}{u^{\gamma}} + \dfrac{(\beta-\pi_{t}^{\kappa}(u))^{2}}{(1-u)^{\gamma}}\right\rbrace du\;dt <\infty$ $\mathbb{Q}$ almost surely. Now we have to prove that the function $ (t,u)\to \pi_{t}^{\kappa} (u) -\alpha$ is in the space $L^{2}(I \times (0,1),dt\otimes d\mu)$, where $ \mu$ is the measure whose density with respect to the Lebesgue measure is given by 
$$u \in (0,1) \to \dfrac{1}{u^{\gamma}}.$$ 
 A similar argument {shows} that the function  $(t,u) \to \pi_{t}^{\kappa}(u) -\beta$ belongs to $L^{2}([0,T]\times (0,1), dt \otimes d\mu')$, where $\mu'$ is the measure whose density with respect to the Lebesgue measure is given by
$$u\in[0,1]\rightarrow\frac{1}{(1-u)^{\gamma}}.$$

Let $\nu _{h}^{N}$ be the Bernoulli product measure corresponding to a profile $h$ which is Lipschitz such that $h(0)= \alpha \leq h(u)\leq \beta =h(1)$ for all $u\in [0,1]$. Let $G\in C_{c}^{\infty}(I\times[0,1])$. Using the entropy and Jensen's inequalities and the Feynman-Kac's formula we get that

\begin{equation}\label{eq:varfor2}
\begin{split}
&{\mathbb E}_{\mu_N} \left[ \int_I  \;  N^{\gamma -1} \sum_{x\in \Lambda_N}  G_{t} r_N^-\Big(\tfrac{x}{N}\Big) (\eta_x(t\Theta(N))-\alpha)\right]dt\\
&\le C_0 +\int_I \sup_{f} \left\{ N^{\gamma-1}  \sum_{x \in \Lambda_N}  (G_{t}r_N^-)\Big(\tfrac{x}{N}\Big) \langle \eta_{x}-\alpha,f\rangle_{\nu^{N}_{h}}+ \Theta(N) N^{-1} \left\langle  L_N   {\sqrt f} , {\sqrt f} \right\rangle_{\nu^{N}_{h}}  \right\}dt,
\end{split}
\end{equation}
where the supremun is  taken over all the densities $f$ on $\Omega_{N}$ with respect to $\nu _{h}^{N}$. 
Using (\ref{dir_est_lip}) with $B=1$ we can bound from above the second term on the right hand side of  (\ref{eq:varfor2}) by 

$$-\dfrac{\Theta(N)}{4N}D_{N}(\sqrt{f},\nu _{h}^{N}) +C \Theta(N)N^{-\gamma}(\kappa N^{-\theta} +1),$$
and from   \ref{bound} with $A_{x} = \tfrac{G_{t}\left(\tfrac{x}{N}\right)}{\kappa}$
the term on the right side of (\ref{eq:varfor2}) is bounded from above by 
 \begin{eqnarray*}
\dfrac{C N^{\gamma-1}}{\kappa}\sum_{x\in\Lambda_N}r_N^-\Big(\tfrac{x}{N}\Big)\left( G_{t}\Big(\tfrac{x}{N}\Big)\right)^{2}+C(\kappa +1).
\end{eqnarray*}
Taking  $N \to \infty$ we can conclude that there exists a constant $C'>0$ independent of $G$ and of $t$ such that
\begin{equation*}
{\bb E}_{\mathbb{Q}} \left[  \int_I \int_0^1\left(\cfrac{(\pi_t^{\kappa}(u) -\alpha) G _{t}(u) }{|u|^{\gamma}}\;  \; -\;   \dfrac{C'}{\kappa} \cfrac{G^2_{t}(u) }{|u|^{\gamma}} \;\right) du dt  \right] \ls\vert I \vert(\kappa +1).
\end{equation*}
From Lemma 7.5 in \cite{KLO2} we can insert the  supremum over $G$ inside the expectation above, and we get
\begin{equation}
\label{2S}
{\bb E}_{\mathbb{Q}} \left[ \sup_G \left\{  \int_I \int_0^1\left(\cfrac{(\pi_t^{\kappa}(u) -\alpha) G _{t}(u) }{|u|^{\gamma}}\; \; -\;   \dfrac{C'}{\kappa} \cfrac{G^2_{t}(u) }{|u|^{\gamma}} \; \right) du dt\right\}  \right] \ls \vert I \vert (\kappa +1) .
\end{equation}
The previous formula implies that 
$$ \int_I\int_0^1 \frac{ (\pi_t^{\kappa}(u)-\alpha)^2}{|u|^\gamma}\, du dt < \infty$$
$\mathbb{Q}$ almost surely. Similarly, we get 
$$ \int_I\int_0^1 \frac{ (\pi_t^{\kappa}(u)-\beta)^2}{|u|^\gamma}\, du dt < \infty$$
$\mathbb{Q}$ almost surely.

\vspace{.5cm}

\noindent\textbf{Final step.}  By Definition \ref{Def. Dirichlet Condition}, the two steps above allow us to show that $\mathbb{Q}$ is concentrated on trajectories  of measures whose density is a weak solution of the corresponding hydrodynamic equation (see  Proposition \ref{prop:weak_sol_car}). By uniqueness of the weak solution (see Lemma \ref{lem:uniquess}) we get that $\mathbb{Q}$ is unique. Indeed, we have that $\mathbb{Q}= \delta_{\{\rho_{t}^{\kappa}(u)du\}}$ (Dirac mass). Then, by using the latter, we compute the expectation in \eqref{1S} and \eqref{2S} and we are done.
  
\begin{flushright}
$\qed$
\end{flushright}
\subsection{Characterization of  limit points}
\label{subsec:Characterization}
In the present subsection we characterize all limit points $\mathbb{Q}$ of the sequence $\lbrace\mathbb{Q} _{N} \rbrace _{N\geq 1}$, { which we know that exist from the results of Subsection \ref{subsec:Tightness}}. Let us assume without lost of generality, that  $\lbrace\mathbb{Q} _{N} \rbrace _{N\geq 1}$ converges to $\mathbb{Q}$. Since there is at most one particle per site, it is easy to show that $\mathbb{Q}$ is concentrated on trajectories {of measures} absolutely continuous with respect to the Lebesgue measure, i.e. $\pi_{t}^{\kappa}(du)=\rho_{t}^{\kappa}(u)du$ (for details see \cite{KL}). In Proposition \ref{prop:weak_sol_car} below we prove, for each range of $\theta$, that  $\mathbb{Q}$ is concentrated on trajectories of measures whose density 
satisfies a weak form of the corresponding hydrodynamic equation. Moreover, we have seen in Theorem \ref{Energy_Thm1} that $\mathbb{Q}$ is concentrated on trajectories of measures whose density satisfies the energy estimate, i.e. $\rho^{\kappa}\in L^{2}(0,T;\mc H^{\gamma/2})$ and
$$\int^{T}_{0}\int_{0}^{1} \left\{ \dfrac{(\alpha - \rho_{t}^{\kappa}(u))^{2}}{u^{\gamma}}+\dfrac{(\beta - \rho_{t}^{\kappa}(u))^{2}}{(1-u)^{\gamma}}\right\}dudt<\infty.$$
Since a weak solution of the hydrodynamic equation \eqref{eq:Dirichlet Equation} is unique we have that $\mathbb{Q}$ is unique and takes the form of a Dirac mass.
\begin{prop}
\label{prop:weak_sol_car}
If $\bb Q$ is a limit point of $ \{\bb Q_{N}\}_{N\geq 1}$  then 
\begin{enumerate}[1.]
\item if $\theta < 0$:
\end{enumerate}$$\bb Q\left( {\pi _{\cdot}}: F_{Reac}(t, \rho^{\kappa},G,g)= 0, \forall t\in [0,T],\, \forall G \in C_c^{1,2} ([0,T]\times[0,1])\,\right)=1.
$$
\begin{enumerate}[2.]
\item if $\theta = 0$:
\end{enumerate}
$$\bb Q\left( {\pi _{\cdot}}: F_{Dir}(t, \rho^{\kappa},G,g)= 0, \forall t\in [0,T],\, \forall G \in C_c^{1,2} ([0,T]\times[0,1])\,\right)=1.
$$
\end{prop}
\begin{proof}
Note that in order to prove the proposition, it is enough to verify, for $\delta > 0$ and $ G$ in the corresponding space of test functions,  that  
\begin{eqnarray} \nonumber
&&\bb Q\left(\pi _{\cdot}\in \mc D_{\mc M^+}^{T}: \sup_{0\le t \le T} \left\vert F_{\theta}(t,\rho^{\kappa},G,g) \right\vert>\delta\right)=0,
\end{eqnarray}
for each $\theta$, where $F_\theta$ stands for $F_{Reac}$ if $\theta<0$ and $F_{Dir}$ if $\theta=0$ . Indeed, we have that
\begin{equation}
\label{def_F_theta}
\begin{split}
F_{\theta}(t, \rho^{\kappa},G,g)=&\left\langle \rho^{\kappa}_{t},  G_{t} \right\rangle -\left\langle g,   G_{0}\right\rangle - \int_0^t\left\langle \rho^{\kappa}_{s},\Big(\partial_s + \mathbb{1}_{\{\theta = 0\}}\bb L \Big) G_{s}  \right\rangle ds \\
+&\mathbb{1}_{\{\theta\leq 0\}}\kappa \int^{t}_{0} \left\langle \rho_{s}^{\kappa}, G_s \right\rangle_{ V_1 } ds 
-\mathbb{1}_{\{\theta\leq 0\}}\kappa \int^{t}_{0}\left\langle G_s , V_0 \right\rangle\,ds=0.
\end{split}   
\end{equation}

 From here on, in order to simplify notation, we will erase $\pi_\cdot$ from the sets that we have to look at. 

\vspace{0.5cm}
 By definition of $F_{\theta}$ above  we can bound from above the previous probability by the sum of
\begin{equation}\label{RD6}
\bb Q\left(  \sup_{0\le t \le T} \left|F_{\theta}(t,\rho^{\kappa},G,\rho_{0}) \right|>\dfrac{\delta}{2}\right)
\end{equation}
and
\begin{equation*}
\bb Q \left(  \left| \left\langle \rho_{0}-g, G_{0}\right\rangle \right|>\dfrac{\delta}{2}\right).
\end{equation*}
We note that {last probability} is equal to zero since $\mathbb Q$ is a limit point of $\{\mathbb Q_N\}_{N\geq 1}$ and $\mathbb Q_N$ is induced by $\mu_N$ which is associated to $g$. {Now we deal with} \eqref{RD6}. Since for $\theta\leq 0$ the function $G_s$ has compact support included in $(0,1)$ the singularities of $V_0$ and $V_1$ are not present, thus  from Proposition A.3 of \cite{FGN}, the set inside the probability in \eqref{RD6} is an open set in the Skorohod topology. Therefore, from Portmanteau's Theorem we bound \eqref{RD6} from above by
\begin{equation*}
\liminf_{N\to\infty}\,\bb Q_{N}\left( \sup_{0\le t \le T} \left|F_{\theta}(t,\rho^{\kappa},G,\rho_{0}) \right|>\dfrac{\delta}{2}\right).
\end{equation*}
Summing and subtracting $\displaystyle\int_{0}^{t} \Theta(N) L_{N}\langle \pi_{s}^{N},G_{s}\rangle ds$ to the term inside the previous absolute {value}, recalling \eqref{Dynkin'sFormula} and  the definition of $\mathbb Q_N$,  we can bound the previous probability   from above by the sum of the next two terms 
\begin{equation*}
 \bb P_{\mu_{N}} \left(\sup_{0\le t \le T} \left\vert M_{t}^{N}(G) \right\vert>\dfrac{\delta}{4}\right)
\end{equation*}
and
\begin{equation}
\label{CLP2}
\begin{split}
&\bb P_{\mu_{N}}  \left( \sup_{0\le t \le T} \left| \int_{0}^{t} \Theta(N) L_{N}\langle \pi_{s}^{N},G_{s}\rangle ds -\int_0^t\left\langle \pi_{s}^{N},\mathbb{1}_{\{\theta = 0\}}\LL  G_{s} \right\rangle \,ds \right.\right.\\
& +\mathbb{1}_{\{\theta\leq 0\}}{ \kappa}\int^{t}_{0}\left\langle  \rho_s,G_{s} \right\rangle_{V_{1}}\,ds
\left.\left.- \mathbb{1}_{\{\theta\leq 0\}}{ \kappa}\int^{t}_{0}\left\langle  G_{s},V_0 \right\rangle\,ds  \right|>\dfrac{\delta}{4}\right). 
\end{split}
\end{equation}
By Doob's inequality we have that
\begin{equation*}
\begin{split}
&\bb P_{\mu_{N}} \left(\sup_{0\le t \le T} \left\vert M_{t}^{N}(G) \right\vert>\dfrac{\delta}{4}\right) 
\ls   \dfrac{1}{\delta^{2}} \bb E_{\mu _{N}} \left[\int_{0}^{T}\Theta(N)\left[ L_{N} \langle\pi^{N}_{s},G \rangle^{2}- 2\langle\pi^{N}_{s},G \rangle L_{N} \langle\pi^{N}_{s},G \rangle\right]ds \right].
\end{split}
\end{equation*}
In the proof of Proposition \ref{Tightness} we have proved that the term inside the time integral in the previous expression is $\mc{O} (N^{\gamma-2})$. Then, using the fact that $\gamma<2$ we have that last probability vanishes as $N\to\infty$. It remains to prove that (\ref{CLP2}) vanishes as $N\to\infty$. For that purpose, we recall (\ref{gen_action})  and we bound (\ref{CLP2}) from above by the sum of the following terms
\begin{equation}
\label{DC1}
\bb P_{\mu_{N}}  \left(\sup_{0\le t \le T} \left| \int_{0}^{t}\cfrac{\Theta(N)}{N-1} \sum_{x\in \Lambda_N}\mathcal{L}_NG_{s}(\tfrac{x}{N})\eta_x^{N}(s) ds-  \int_{0}^{t}\left\langle \pi_{s}^{N},\mathbb{1}_{\{\theta= 0\}}\LL G_{s} \right\rangle   \, ds \right|>\dfrac{\delta}{2^{4}}\right),
\end{equation}
\begin{multline}
\label{DC2}
\bb P_{\mu_{N}}\left(\sup_{0\le t \le T} \left|   \int_{0}^{t} \left\{ \dfrac{ \kappa \Theta(N)}{N^{\theta}(N-1)} \sum_{x \in \Lambda_N}  (G_{s} r_{N}^{-})(\tfrac{x}{N}){(\alpha-\eta_x^{N}(s))}  \right.\right.\right.\\
  \left.\left. \left.- \mathbb{1}_{\{\theta\leq 0\}} \kappa \int_{0}^{1}  (G_{s}r^{-})(u)(\alpha - \rho_{s}^{\kappa}(u))du \right\}\, ds\right | > \dfrac{\delta}{2^{4}}\right)
\end{multline}
and
\begin{multline}
\label{DC3}
\bb P_{\mu_{N}}\left(\sup_{0\le t \le T} \left|   \int_{0}^{t} \left\{ \dfrac{ \kappa \Theta(N)}{N^{\theta}(N-1)} \sum_{x \in \Lambda_N}  (G_{s} r_{N}^{+})(\tfrac{x}{N}){(\beta-\eta_x^{N}(s))} \right.\right.\right.\\
  \left. \left.\left.- \mathbb{1}_{\{\theta\leq 0\}}\kappa \int_{0}^{1}  (G_{s}r^{+})(u)(\beta - \rho_{s}^{\kappa}(u))du \right\}\, ds\right | > \dfrac{\delta}{2^{4}}\right).
\end{multline}
For $\theta=0$   from  \eqref{F_convergenceLL}  we have that (\ref{DC1}) goes to $0$ as $N\to \infty$. For $\theta\leq 0$ we have that from  \eqref{F_convergenceLL} and \ref{F_convergencer} the boundary terms   (\ref{DC2}) and (\ref{DC3}) go to $0$ as $N\to \infty$. This finishes the proof Proposition \ref{prop:weak_sol_car}. 
\end{proof}

\section{Proof of Theorem \ref{convergence_rho^k_to_rho^0}} 
\label{sec: Study of solution}

For easy understanding of the proof of items i) and ii) of Theorem \ref{convergence_rho^k_to_rho^0}, we first {establish} some notation and prove some lemmata.

Recall the function $\bar\rho^{\infty}$ introduced  in Remark \ref{rem:Explicit_sol_rho_infty} which can be {rewritten} as 
$$\bar\rho^{\infty}(u)=\dfrac{ \beta u^{\gamma} + \alpha(1-u)^\gamma}{u^{\gamma}+(1-u)^{\gamma}}. $$
It is easy to see that $\bar\rho^{\infty}(0)=\alpha$ and $\bar\rho^{\infty}(1)=\beta$. {Moreover, it is not difficult to see that $\bar\rho^{\infty} \in C^{1}([0,1])$ and that
$$\lim _{u\to 0} (\bar\rho^{\infty}(u))^{\prime}u^{2-\gamma} = \lim _{u\to 1} (\bar\rho^{\infty}(u))^{\prime}(1-u)^{2-\gamma} = 0,$$ and from Lemma 7.2 of \cite{GM} we conclude that
\begin{equation}\label{bar_rho^infty_in_H}
\Vert \bar\rho^{\infty} \Vert_{\gamma/2}<\infty.
\end{equation}}

By the fractional Hardy's inequality (see e.g. \cite{Dy}) and the fact that $V_{1}(\tfrac{1}{2})\leq V_{1}(u)$ for all $u\in (0,1)$ we know that 
\begin{equation} \label{norms_related}
\Vert g \Vert \ls \Vert g \Vert _{V_{1}}\ls \Vert g\Vert_{\gamma/2}
\end{equation}
for any $g\in \mc H^{\gamma/2}_{0}$. 

In order to prove items i) and ii) of Theorem \ref{convergence_rho^k_to_rho^0} we first guarantee the existence of weak solutions of equation (\ref{eq:Dirichlet Equation}) with $\kappa=0$ and \eqref{eq:Dirichlet Equation_infty}, (see Lemma \ref{lem1_existence_rho^0} and \ref{lem1_existence_rho^infty} below), then we establish the convergence in $L^{2}(0,T;L^{2})$ (see Lemma \ref{lem2_rho^kappa_to_rho^0_L^2} and \ref{lem2_rho^kappa_to_rho^infty_L^2}) which will allow us to conclude.
 
\begin{lem} \label{lem1_existence_rho^0}
Let $\rho_{0}:[0,1]\to[0,1]$ be a measurable function. Then, there exists a weak solution of (\ref{eq:Dirichlet Equation})  with $\hat\kappa =0$ and  initial condition $\rho_{0}$.
\begin{proof}
The strategy of the proof is to construct the solution as  the limit of $\rho^{\kappa}$, {as $k\to 0$, where $\rho^k$ is the} weak solution of (\ref{eq:Dirichlet Equation}) with initial condition $\rho_{0}$ and $\hat\kappa = \kappa$.

By item i) in Theorem \ref{Energy_Thm1} and since $\kappa>0$ we know  that  
\begin{equation}\label{EE00}
\int_{I}\Vert \rho_{t}^{\kappa} \Vert_{\gamma/2}^{2}dt \lesssim |I|(\kappa +1) 
\end{equation}
for any interval $I\subset[0,T]$. We define 
\begin{equation}\label{dvf}
\forall t \in [0,T], \quad \forall u \in [0,1], \quad \vf^{\kappa}_{t}(u):= \rho^{\kappa}_{t}(u) - \bar\rho^{\infty}(u).
\end{equation}
Since we are interested in small values of $\kappa$, say $\kappa \leq 1$, from \eqref{EE00}, (\ref{bar_rho^infty_in_H}) and the fact $(a +b)^{2}\leq 2a^{2}+2b^{2}$, it is not difficult to see that
\begin{equation}\label{EE}
\int_{I}\Vert \vf_{t}^{\kappa} \Vert_{\gamma/2}^{2}dt \lesssim |I|,
\end{equation}
thus we have that $ \vf^{\kappa}\in L^{2}(0,T;\mc H ^{\gamma/2}_{0}).$
It is also easy to see that $\vf^{\kappa}$ satisfies
\begin{multline}\label{vf^k}
\langle \vf_{t}^{\kappa} , G_{t} \rangle  - \langle \vf_{0}, G_{0} \rangle 
- \int_0^t\left\langle \vf_{s}^{\kappa} ,\left(\LL  + \partial_s\right) G_{s} \right\rangle \, ds  + {\kappa}\int_0^t\langle \vf_{s}^{\kappa},G_s\rangle_{V_{1}} ds-\int_0^t\langle \bar\rho^{\infty},\LL G_s\rangle ds=0   
\end{multline}
for all $t\in [0,T]$, for any function $G \in C^{1,\infty}_{c}([0,T]\times(0,1))$ and where $\vf_{0}(u)= \rho_{0}(u) - \bar\rho^{\infty}(u)$. { From \eqref{EE} we conclude that} there exists a subsequence of $(\vf^{\kappa})_{\kappa\in (0,1)}$ converging weakly to some element $\vf^{0}\in L^{2}(0,T;\mc H^{\gamma/2}_{0})$ as $\kappa\to 0$.  We claim that $\rho^{0}:= \bar\rho^{\infty} + \vf^{0}$ is the desired solution. Indeed, first note that since the norm $\Vert\cdot \Vert_{\gamma/2}$ is weakly lower-semicontinuous we have that 
\begin{equation}\label{EE0}
\int_{I}\Vert \vf_{t}^{0} \Vert_{\gamma/2}^{2}dt \lesssim |I|.
\end{equation}
 By using $(a +b)^{2}\leq 2a^{2}+2b^{2}$ we have that 
\begin{equation*}
\begin{split}
\int_{I}\Vert \rho^{0}_{t} \Vert_{\gamma/2}^{2}dt &\leq 2\int_{I}\Vert \bar\rho^\infty \Vert_{\gamma/2}^{2}dt +2\int_{I}\Vert \vf^{0}_{t} \Vert_{\gamma/2}^{2}dt \lesssim|I|.
\end{split}
\end{equation*}
Taking $I = [0,T]$, we have  that $\rho^{0}$ satisfies item i) of Definition \ref{Def. Dirichlet Condition}. Since $\vf^{0}\in L^{2}(0,T;\mc H^{\gamma/2}_{0})$, it is easy to see that $\rho^{0}_{t}(0)= \bar\rho^{\infty}(0)=\alpha$ and $\rho^{0}_{t}(1)=\bar \rho^{\infty}(1)=\beta$ for almost every $t \in [0,T]$. Then, item ii) for $\hat\kappa=0$ in Definition \ref{Def. Dirichlet Condition} is satisfied. In order to verify that $\rho^{0}$ satisfies item iii) in Definition \ref{Def. Dirichlet Condition} we first integrate (\ref{vf^k}) over $[0,t]$. Thus we have that
\begin{equation*}
\begin{split}
&\int_{0}^{t}\langle \vf_{s}^{\kappa} , G_{s} \rangle ds  - t\langle \vf_{0}, G_{0} \rangle 
- \int_{0}^{t}\int_0^s\left\langle \vf_{r}^{\kappa} ,\left(\LL  + \partial_r\right) G_{r} \right\rangle \, drds\\  +& {\kappa}\int_{0}^{t}\int_0^s\langle \vf_{r}^{\kappa},G_r\rangle_{V_{1}} drds-\int_{0}^{t}\int_0^s\langle \bar\rho^{\infty},\LL G_r\rangle drds=0   
\end{split}
\end{equation*}
for any function $G \in C^{1,\infty}_{c}([0,T]\times(0,1)).$
Taking $\kappa \to 0$, by weak convergence and Lebesgue's dominated convergence theorem  we get from the previous equality that 
\begin{multline*}
\int_{0}^{t}\langle \vf_{s}^{0} , G_{s} \rangle ds  - t\langle \vf_{0}, G_{0} \rangle 
- \int_{0}^{t}\int_0^s\left\langle \vf_{r}^{0} ,\left(\LL  + \partial_r\right) G_{r} \right\rangle  -\langle \bar\rho^{\infty},\LL G_r\rangle drds =0.   
\end{multline*}
Now, taking the derivative with respect {to $t$ in the previous equality} we get that $\vf^{0}$ satisfies 
\begin{equation}\label{vf^0}
\begin{split}
&\langle \vf_{t}^{0} , G_{t} \rangle  - \langle \vf_{0}, G_{0} \rangle 
- \int_0^t\langle \vf_{s}^{0} ,\Big(\LL  + \partial_s\Big) G_{s} \rangle \, ds -\int_0^t\langle \bar\rho^{\infty},\LL G_s\rangle ds =0,
\end{split}   
\end{equation}
 for all $t \in [0,T]$. Then,  item iii) with $\kappa = 0$ in Definition \ref{Def. Dirichlet Condition} follows from (\ref{vf^0}), the definition of $\rho^{0}$ and $\bar\rho^{\infty}$
\end{proof}
\end{lem}

\begin{lem}\label{lem2_rho^kappa_to_rho^0_L^2} Let $\rho_{0}:[0,1]\to[0,1]$ be a measurable function. Let $\rho^{\kappa}$ {be} the weak solution of (\ref{eq:Dirichlet Equation})  with initial condition $\rho_{0}$ and $\hat \kappa = \kappa$. Then, $\rho^{\kappa}$ converges strongly to $\rho^{0}$ in $L^{2}(0,T;L^{2})$ as $\kappa$ goes to $0$, where $\rho^{0}$ is the weak solution of (\ref{eq:Dirichlet Equation}) with $\hat\kappa = 0$ and initial condition $\rho_{0}$. \begin{proof}
Note that is enough to show that 
\begin{equation*}
\int_0^t  \Vert \rho^{\kappa}_{s}-\rho^{0}_{s} \Vert^{2} \,ds \lesssim t^{2}{\kappa},
\end{equation*}
for all $t\in [0,T]$. By Lemma \ref{lem1_existence_rho^0} we know that  $\rho^{0} = \bar\rho^{\infty} + \vf^{0}$. {Then,}  last inequality is equivalent to 
\begin{equation}\label{Imp}
\int_0^t  \Vert \vf_{s}^{\kappa} -\vf_{s}^{0} \Vert^{2} \,ds \lesssim  t^{2}{\kappa}.
\end{equation}
 By subtracting  (\ref{vf^0}) from (\ref{vf^k}) and calling  $\delta_t^k:=\vf_{t}^{\kappa} -\vf_{t}^{0}$ we obtain that
\begin{equation}\label{vf^k-vf^0}
\begin{split}
&\langle \delta_{t}^{\kappa}, G_{t} \rangle 
- \int_0^t \left\langle \delta_{s}^{\kappa},\left(\LL  + \partial_s\right) G_{s} \right\rangle \, ds  =- {\kappa}\int_0^t\langle \vf_{s}^{\kappa},G_s\rangle_{V_{1}} ds
\end{split}   
\end{equation}
for any function $G \in C^{1,\infty}_{c}([0,T]\times(0,1))$. Let $\lbrace H_{n}^{\kappa}\rbrace_{n\geq 1}$ be a sequence of functions in $C_{c}^{1,\infty}([0,T]\times(0,1))$ converging to $\delta^{\kappa}$ as $n\to \infty$ with respect to the norm of $L^{2}(0,T;\mc H ^{\gamma/2}_{0})$ and {for $n\geq 1$,} let $G_{n}^{\kappa}(s,u) = \int^{t}_{s} H_{n}^{\kappa}(r,u)dr $. We claim that by plugging $G_{n}$ into (\ref{vf^k-vf^0}) and taking $n\to \infty$ we get that 
\begin{equation} \label{Approx1}
\begin{split}
& 
 \int_{0}^t \Vert \delta_{s}^{\kappa} \Vert^{2} \, ds  +  \dfrac{1}{2}  \left \Vert \int_{0}^{t} \delta_{s}^{\kappa}ds \right\Vert_{\gamma/2}^{2}  = {-\kappa} \int_{0}^t\left\langle \vf_{s}^{\kappa}, \int_{s}^{t} \delta_{r}^{\kappa} dr\right\rangle_{V_{1}} ds.
\end{split}   
\end{equation}
We leave the justification of the equality above to the end of the proof. Now, by using successively the Cauchy-Schwarz's inequality we have that 
\begin{equation} \label{B1}
\begin{split}
\int_{0}^t \Vert \delta_{s}^{\kappa} \Vert^{2} \, ds  +  \dfrac{1}{2} \left \Vert \int_{0}^{t} \delta_{s}^{\kappa}ds \right\Vert_{\gamma/2}^{2}  
&\leq {\kappa}\int_{0}^t \Vert \vf_{s}^{\kappa}\Vert_{V_{1}} \left\Vert \int_{s}^{t} \delta_{r}^{\kappa} dr\right\Vert_{V_{1}} ds  
\\ &\lesssim  {\kappa} \sqrt{\int_{0}^t \Vert \vf_{s}^{\kappa}\Vert^{2}_{\gamma/2} ds}\sqrt{\int_{0}^t\left\Vert \int_{s}^{t} \delta_{r}^{\kappa} dr\right\Vert_{\gamma/2}^{2} ds}.
\end{split}   
\end{equation}
In the last inequality of the previous expression we used (\ref{norms_related}). By the triangular inequality we have that $\sqrt{\int_{0}^t\Big\Vert \int_{s}^{t} \delta_{r}^{\kappa} dr\Big\Vert_{\gamma/2}^{2} ds}$ is bounded from above by
\begin{equation}\label{B2}
\begin{split}
\sqrt{\int_{0}^t \left(\int_{s}^{t} \Vert\delta_{r}^{\kappa}\Vert_{\gamma/2}dr\right)^{2} ds} \leq  \sqrt{t\int_{0}^t \int_{0}^{t} \Vert\delta_{r}^{\kappa}\Vert_{\gamma/2}^{2}dr ds}
\lesssim \sqrt{t^{2} \int_{0}^{t} \left(\Vert\vf_{r}^{\kappa}\Vert_{\gamma/2}^{2} + \Vert\vf_{r}^{0}\Vert_{\gamma/2}^{2} \right) dr }.
\end{split}
\end{equation}
In the first inequality in the previous display we used the Cauchy-Schwarz's inequality and in the second inequality we used the Minkowski's inequality and the inequality $(a +b)^{2}\leq 2(a^{2}+b^{2})$. Using (\ref{EE}) and (\ref{EE0}), we get from (\ref{B1}) and (\ref{B2})  the result. 

We conclude this proof  justifying \eqref{Approx1}. Note that it is enough to show 
\begin{enumerate}[i)]
\item $\displaystyle\lim_{n \to \infty}  \int_0^t \langle \delta_{s}^{\kappa}, (\partial_s G_n^{\kappa}) (s,\cdot)\rangle  ds = - \int_0^t  \Vert {\delta}^{\kappa}_{s}\Vert^{2} ds$. 
\item $\displaystyle\lim_{n \to \infty} \int_0^t \langle { \delta}_{s}^{\kappa},  \bb L G_{n}^{\kappa} (s,\cdot)  \rangle ds =- \cfrac{1}{2} \; \Big\| \int_0^t {\delta}_s^{\kappa} ds \Big\|^2_{\gamma/2}.$ 
\item $\displaystyle\lim_{n \to \infty} \int^{t}_{0}  \left\langle 
\vf_{s}^{\kappa}, G_n^{\kappa} (s,\cdot) \right \rangle_{V_{1}} ds =\int^{t}_{0}  \left\langle 
\vf_{s}^{\kappa}, \int_{s}^{t}\delta_{r}^{\kappa}dr\right \rangle_{V_{1}} ds.$
\end{enumerate}
For i) we rewrite  $\int_0^t \langle \delta_{s}^{\kappa}, (\partial_s G_n^{\kappa}) (s,\cdot)\rangle  ds$ as 
\begin{equation*}
\begin{split}
 & - \int_0^t  \langle {\delta}_s^{\kappa} \, , \, H_n^{\kappa} (s, \cdot) \rangle \, ds = -\int_0^t  \big\langle {\delta}_s^{\kappa} \, , \, H_n ^{\kappa}(s, \cdot) - { \delta}_s^{\kappa} \big\rangle \, ds 
 - \; \int_0^t \| {\delta}_s ^{\kappa}\|^2 \, ds.
\end{split}
\end{equation*}
Observe then that by the Cauchy-Schwarz's inequality we have
\begin{equation*}
\begin{split}
& \left| \int_0^T  \big\langle {\delta}_s^{\kappa} \, , \, H_n^{\kappa} (s, \cdot) - {\delta}_s ^{\kappa}\big\rangle \, ds \right| \le \int_0^T \| {\delta}_s^{\kappa} \| \, \| H_n ^{\kappa}(s, \cdot) - {\delta}_s^{\kappa} \| \, ds\\
&\le \sqrt{ \int_0^T  \| {\delta}_s^{\kappa} \|^2 \, ds} \; \sqrt{ \int_0^T  \| H_n ^{\kappa}(s, \cdot) - {\delta}_s^{\kappa} \|^2 \, ds } 
\end{split} 
\end{equation*}
which goes to $0$ as $n \to \infty$ since $H_{n}^{\kappa}\to \delta_{s}^{\kappa}$ in $L^{2}(0,T;\mc H^{\gamma/2}_{0})$. For ii), since $G_{n}$ has compact support included in $(0,1)$, we can use the integration by parts formula for the regional fractional Laplacian (see Theorem 3.3 in \cite{GM}) which permits to write
$$ \int_0^t \langle {\delta}_{s}^{\kappa} , \bb L G_{n}^{\kappa}(s,\cdot)\rangle ds = - \int_0^t \Big\langle \delta_s^{\kappa} , G_n^{\kappa} (s, \cdot) \, \Big\rangle_{\gamma/2}\, ds.$$ 
Then we have
\begin{equation*}
\begin{split}
&\int_0^t \Big\langle \delta_s ^{\kappa}\,  , G_n^{\kappa} (s, \cdot) \, \Big\rangle_{\gamma/2}\, ds = \int_0^t \Big\langle \delta_s ^{\kappa}\,  , \int_s^t {\delta}_r^{\kappa} dr  \, \Big\rangle_{\gamma/2}\, ds + \int_0^t \Big\langle \delta_s^{\kappa} \,  , G_n ^{\kappa}(s, \cdot) - \int_s^t {\delta}_r^{\kappa} dr \, \Big\rangle_{\gamma/2}\, ds\\
&=\iint_{0 \le s < r \le t} \langle {\delta}_s^{\kappa} \, , \, {\delta}_r^{\kappa} \rangle_{\gamma/2} \, ds dr \; + \;  \int_0^t \Big\langle \delta_s^{\kappa} \,  ,  \int_s^t \left( H_n ^{\kappa}(r, \cdot) -{\delta}_r^{\kappa}\right)  dr \, \Big\rangle_{\gamma/2}\, ds\\
&= \cfrac{1}{2} \, \iint_{[0,t]^2} \langle {\delta}_s^{\kappa} \, , \, {\delta}_r^{\kappa} \rangle_{\gamma/2} \, ds dr \; + \;  \int_0^t \Big\langle \delta_s ^{\kappa}\,  , \int_s^t \left( H_n^{\kappa} (r, \cdot) -{\delta}_r^{\kappa}\right)  dr \, \Big\rangle_{\gamma/2}\, ds\\
&=\cfrac{1}{2} \; \Big\| \int_0^t {\delta}_s ^{\kappa}ds \Big\|^2_{\gamma/2}+ \;  \int_0^t \Big\langle \delta_s^{\kappa} \,  , \int_s^t \left( H_n ^{\kappa}(r, \cdot) -{\delta}_r^{\kappa}\right)  dr \, \Big\rangle_{\gamma/2}\, ds.
\end{split}
\end{equation*}
To conclude the proof of ii) it is sufficient to show that the term at  the right hand side of last expression vanishes as $n$ goes to $\infty$. This is a consequence of a successive use of Cauchy-Schwarz's inequalities:
\begin{equation}\label{Con1}
\begin{split}
&\left| \int_0^t \Big\langle \delta_s ^{\kappa}\,  , \int_s^t \left( H_n^{\kappa} (r, \cdot) -{\delta}_r^{\kappa}\right)  dr \, \Big\rangle_{\gamma/2}\, ds \right| \le \int_0^t \Big\| \delta_s^{\kappa} \Big\|_{\gamma/2} \; \Big\| \int_s^t \left( H_n^{\kappa} (r, \cdot) -{\delta}_r^{\kappa} \right) dr \Big\|_{\gamma/2}\, ds\\
&\le \int_0^t \Big\| \delta_s^{\kappa} \Big\|_{\gamma/2} \; \int_s^t \Big\| H_n^{\kappa} (r, \cdot) -{\delta}_r^{\kappa} \Big\|_{\gamma/2}\, dr \,  ds  \le \int_0^t \Big\| \delta_s^{\kappa} \Big\|_{\gamma/2} \; \int_0^t \Big\|  H_n^{\kappa} (r, \cdot) -{\delta}_r^{\kappa} \Big\|_{\gamma/2}\, dr \,  ds \\
&= \left(\int_0^t \Big\|\delta_s^{\kappa} \Big\|_{\gamma/2} ds \right) \, \left(  \int_0^t \Big\|  H_n^{\kappa} (r, \cdot) -{\delta}_r^{\kappa}  \Big\|_{\gamma/2}\, dr \right)\\
& \le t \, \sqrt{ \int_0^t \Big\| \delta_s^{\kappa} \Big\|^2_{\gamma/2} ds} \; \sqrt{  \int_0^t \Big\|  H_n ^{\kappa}(r, \cdot) -{\delta}_r^{\kappa} \Big\|_{\gamma/2}^2\, dr} \; \xrightarrow[n \to \infty]{} \; 0.
\end{split}
\end{equation}
To prove iii) we rewrite  $\int_0^t \langle \vf_{s}^{\kappa},  G_n^{\kappa} (s,\cdot)\rangle_{V_{1}}  ds$ as 
\begin{equation*}
\begin{split}
 &  \int_0^t  \left\langle {\vf}_s^{\kappa} \, , \int^{t}_{s} \left( H_n^{\kappa} (r, \cdot)-\delta_{r}^{\kappa}\right) dr\right\rangle_{V_{1}} \, ds +\int_0^t  \left\langle {\vf}_s^{\kappa} \, , \int^{t}_{s}\delta_{r}^{\kappa}dr \right\rangle_{V_{1}} \, ds 
\end{split}
\end{equation*}
and, to conclude the proof it is sufficient to show that the term at the left hand side of last expression vanishes as $n\to\infty$.
This is a consequence of a successive use of the Cauchy-Schwarz's inequality as in \eqref{Con1}, with $\Vert \cdot\Vert_{\gamma/2}$ replaced by $\Vert \cdot\Vert_{V_{1}}$ and Hardy's inequality:
\begin{equation*}
\begin{split}
&\left| \int_0^t \Big\langle \vf_s^{\kappa} \,  , \int_s^t \{ H_n^{\kappa} (r, \cdot) -{\delta}_r^{\kappa}\}  dr \, \Big\rangle_{V_{1}}\, ds \right| \le \int_0^t\Big\| \vf_s^{\kappa} \Big\|_{V_{1}} \; \Big\| \int_s^t \left( H_n ^{\kappa}(r \cdot) -{\delta}_r^{\kappa} \right) dr \Big\|_{V_{1}}\, ds\\
&\le \int_0^t \Big\| \vf_s^{\kappa} \Big\|_{V_{1}} \; \int_s^t \Big\| H_n^{\kappa} (r, \cdot) -{\delta}_r^{\kappa}  \Big\|_{V_{1}}\, dr \,  ds  \le \int_0^t \Big\| \vf_s ^{\kappa}\Big\|_{V_{1}} \; \int_0^t \Big\|  H_n ^{\kappa}(r, \cdot) -{\delta}_r^{\kappa}  \Big\|_{V_{1}}\, dr \,  ds \\
&= \left(\int_0^t \Big\| \vf_s^{\kappa} \Big\|_{V_{1}} ds \right) \, \left(  \int_0^t \Big\|  H_n^{\kappa} (r, \cdot) -{\delta}_r^{\kappa}  \Big\|_{V_{1}}\, dr \right)\\
& \le t \, \sqrt{ \int_0^t \Big\| \vf_s^{\kappa} \Big\|^2_{V_{1}} ds} \; \sqrt{ \int_0^t \Big\|  H_n ^{\kappa}(r, \cdot) -{\delta}_r^{\kappa}  \Big\|_{V_{1}}^2\, dr} \\
& \le C t \, \sqrt{ \int_0^t \Big\| \vf_s ^{\kappa}\Big\|^2_{\gamma/2} ds} \; \sqrt{  \int_0^t \Big\|  H_n ^{\kappa}(r, \cdot) -{\delta}_r^{\kappa} \Big\|_{\gamma/2}^2\, dr} \; \xrightarrow[n \to \infty]{} \; 0
\end{split}
\end{equation*}
{where in the last inequality} we used the fractional Hardy's inequality (see \eqref{norms_related}).

\end{proof}
\end{lem}

\begin{lem} \label{lem1_existence_rho^infty}
Let $\rho_{0}:[0,1]\to[0,1]$ be a measurable function. Consider the function 
$\rho^{\infty}_{t} = \bar \rho^{\infty} + (\rho_{0} -\bar\rho^{\infty})e^{-tV_{1}}$. 
If $g^{\infty}:=\rho_{0} - \bar\rho^{\infty} \in \mc H^{\gamma/2}$, then
\begin{enumerate}[i)]
\item  $\rho^{\infty} \in L^{2}(0,T;\mc H^{\gamma/2})$ .
\item $\rho^{\infty}$ is a weak solution of (\ref{eq:Dirichlet Equation_infty}) with initial condition $\rho_{0}$.
\end{enumerate}
\begin{proof}
For $i)$ note that by using the inequality $(a+b)^2\leq 2a^{2}+2b^{2}$ we get that
\begin{equation*}
 \int^{T}_{0}\Vert \rho_{t}^{\infty} \Vert_{\gamma/2}^{2}dt \leq   2T\Vert \bar\rho^{\infty} \Vert_{\gamma/2}^{2} + 2\int_{0}^{T} \left \Vert g^{\infty}e^{-tV_{1}} \right\Vert_{\gamma/2}^{2}dt.
\end{equation*}
Since $\Vert \bar\rho^{\infty} \Vert_{\gamma/2}<\infty$ ({see \eqref{bar_rho^infty_in_H}})  it is enough to prove that the term on the right hand side of {last expression} is finite. Note that  $\left \Vert g^{\infty}e^{-tV_{1}} \right\Vert_{\gamma/2}^{2}$ is equal to 
\begin{equation*}
\begin{split}
 &\dfrac{c_{\gamma}}{2}\iint_{[0,1]^{2}}\dfrac{\left( g^{\infty}(u)e^{-tV_{1}(u)} - g^{\infty}(v)e^{-tV_{1}(v)}\right)^{2}}{\vert u-v\vert ^{\gamma+1}}\,dudv\\
 &= \dfrac{c_{\gamma}}{2}\iint_{[0,1]^{2}}\dfrac{\left( g^{\infty}(u)\left(e^{-tV_{1}(u)}-e^{-tV_{1}(v)}\right) + \left(g^{\infty}(u)- g^{\infty}(v)\right) e^{-tV_{1}(v)}\right)^{2}}{\vert u-v\vert ^{\gamma+1}}\,dudv.
\end{split}
\end{equation*}
 Using the fact that $(a+b)^{2} \leq 2a^{2}+2b^{2}$ and that $\vert g^{\infty}(u)\vert \leq 2 $ for any $u\in [0,1]$  we get that last expression is less than $
8\Vert e^{-tV_{1}}\Vert_{\gamma/2}^{2} + 2\Vert g^{\infty}\Vert_{\gamma/2}^{2}.$ Note that the term $8\Vert e^{-tV_{1}}\Vert_{\gamma/2}^{2}$ can be written as
\begin{equation*}
\begin{split}
&4c_{\gamma}\iint_{[0,1]^{2}}\dfrac{\left(\int_{v}^{u} -t V_{1}'(w)e^{-tV_{1}(w)}dw \right)^{2}}{\vert u - v \vert^{\gamma+1}}dudv\\
=&4c_{\gamma}\iint_{[0,1]^{2}}\dfrac{\left(\int_{v}^{u} t \left(\tfrac{\gamma}{w}r^{-}(w) -\tfrac{\gamma}{1-w}r^{+}(w) \right) e^{-tV_{1}(w)}dw \right)^{2}}{\vert u - v \vert^{\gamma+1}}dudv.
\end{split}
\end{equation*}
 Using again $(a+b)^{2} \leq 2a^{2}+2b^{2}$ and the fact that  $e^{-tV_{1}(w)}\leq e^{-tr^{\pm}(w)}$ for any $w \in [0,1]$, we get that last expression is bounded from above by 
\begin{equation*}
\begin{split}
& 8c_{\gamma}\iint_{[0,1]^{2}}\dfrac{\left(\int_{v}^{u} \tfrac{\gamma}{w} tr^{-}(w)e^{-tr^{-}(w)}dw \right)^{2}}{\vert u - v \vert^{\gamma+1}}+\dfrac{\left(\int_{v}^{u} \tfrac{\gamma}{1-w} tr^{+}(w)e^{-tr^{+}(w)}dw \right)^{2}}{\vert u - v \vert^{\gamma+1}}dudv \\
=&  16c_{\gamma}\iint _{[0,1]^{2}}\dfrac{\left(\int_{v}^{u} \tfrac{\gamma}{w} tr^{-}(w)e^{-tr^{-}(w)}dw \right)^{2}}{\vert u - v \vert^{\gamma+1}}dudv .
\end{split}
\end{equation*}
In the last equality we used a symmetry argument. We can write {last expression} as 
\begin{equation*}
\begin{split}
& C_{\gamma} t^{\tfrac{2-2\gamma}{\gamma}}\iint _{[0,1]^{2}}\dfrac{ \Big(\int_{v}^{u} w^{\gamma-2} (tr^{-}(w))^{\tfrac{2\gamma-1}{\gamma}}e^{-tr^{-}(w)}dw \Big)^{2}}{\vert u - v \vert^{\gamma+1}}dudv, \\
\end{split}
\end{equation*}
where $C_{\gamma} = 16 c_{\gamma}^{\tfrac{2-\gamma}{\gamma}}\gamma^{\tfrac{4\gamma-2}{\gamma}}$.  Since the function  $E_{\gamma}:[0,\infty)\to [0,\infty) $ defined as $E_{\gamma}(z)= z^{\tfrac{2\gamma-1}{\gamma}}e^{-z}$  is bounded from above by  $E_{\gamma}\left(\tfrac{2\gamma-1}{\gamma}\right)$ we can bound {last expression} from above by 
\begin{equation*}
\begin{split}
& C_{\gamma} t^{\tfrac{2-2\gamma}{\gamma}}E^{2}_{\gamma}(\tfrac{2\gamma-1}{\gamma})\iint_{[0,1]^{2}}\dfrac{\left(\int_{v}^{u} w^{\gamma-2} dw \right)^{2}}{\vert u - v \vert^{\gamma+1}}dudv\\ &= C_{\gamma} t^{\tfrac{2-2\gamma}{\gamma}}E^{2}_{\gamma}(\tfrac{2\gamma-1}{\gamma})  (\gamma -2)^{-2}\iint_{[0,1]^{2}}\dfrac{\left( u^{\gamma-1} - v^{\gamma-1} \right)^{2}}{\vert u - v \vert^{\gamma+1}}dudv,
 \end{split}
\end{equation*}
which is finite from (7.2) in the proof of Lemma 7.2 of \cite{GM}. Thus, we have that 
\begin{equation}\label{eq:rhoinfty8}
8\Vert e^{-tV_{1}}\Vert^{2}_{\gamma/2} \lesssim t^{\tfrac{2-2\gamma}{\gamma}}.
\end{equation}
Therefore, if $g^{\infty}\in \mc H^{\gamma/2}$ then we conclude that 
\begin{equation*}
\begin{split}
 \int^{T}_{0}\Vert \rho_{t}^{\infty} \Vert_{\gamma/2}^{2}dt &\lesssim   T \Vert \bar\rho^{\infty} \Vert_{\gamma/2}^{2} + T \left \Vert g^{\infty} \right\Vert_{\gamma/2}^{2} + \int_{0}^{T} t^{\tfrac{2-2\gamma}{\gamma}}dt\\
&\lesssim T \Vert \bar\rho^{\infty} \Vert_{\gamma/2}^{2} + T \left \Vert g^{\infty} \right\Vert_{\gamma/2}^{2} + T^{\tfrac{2-\gamma}{\gamma}} ,
\end{split}
\end{equation*}
which is finite since $\gamma < 2$. 

For $ii)$, since $\rho ^{\infty}$ is the solution of (\ref{eq:Dirichlet Equation_infty}) {then it satisfies} item $ii)$ of Definition \ref{Def. Dirichlet Condition_kappa^infty}. In order to see that $\rho^{\infty}$ satisfies item $i)$ of Definition \ref{Def. Dirichlet Condition_kappa^infty}, note that using $(a + b)^{2}\leq 2a^{2}+2b^{2}$ we have that 
\begin{equation*}
\begin{split}
&\int_{0}^{T}\int^{1}_{0}\left(\dfrac{\left( \alpha - \rho^{\infty}_{t}(u) \right)^{2}}{u^{\gamma}}+\dfrac{\left( \beta- \rho^{\infty}_{t}(u) \right)^{2}}{(1-u)^{\gamma}}\right) dudt\\
&\leq 2T\int^{1}_{0}\left(\dfrac{\left( \alpha - \bar\rho^{\infty}(u) \right)^{2}}{u^{\gamma}}+\dfrac{\left( \beta- \bar\rho^{\infty}(u) \right)^{2}}{(1-u)^{\gamma}}\right) du + \dfrac{8\gamma}{c_{\gamma}}\int_{0}^{T}\Vert e^{-tV_{1}} \Vert_{V_{1}}^{2}dt\\
&= 2T(\beta -\alpha)^{2}\int^{1}_{0}\left( u^{\gamma} +(1-u)^{\gamma}\right) du + \dfrac{8\gamma}{c_{\gamma}}\int_{0}^{T}\Vert e^{-tV_{1}} \Vert_{V_{1}}^{2}dt\\
&\leq 2^\gamma(\beta -\alpha)^{2}T+ \dfrac{8\gamma}{c_{\gamma}}\int_{0}^{T}\Vert e^{-tV_{1}} \Vert_{V_{1}}^{2}dt.
\end{split}
\end{equation*}
For the term on the right hand side of last expression we first see that we can extend continuously the function $e^{-tV_{1}}$ in such a way that it vanishes at $0$ and at $1$.  There exists a constant $C_{2}$ (see \ref{norms_related}) such that  the previous expression is bounded from above by
\begin{equation}
 2^\gamma(\beta -\alpha)^{2}T+ \dfrac{8\gamma C_{2}^{2}}{c_{\gamma}}\int_{0}^{T}\Vert e^{-tV_{1}} \Vert_{\gamma/2}^{2}dt.
\end{equation}
 Thus, we obtain the desired result by using (\ref{eq:rhoinfty8}). 
\end{proof}
\end{lem}

\begin{lem}\label{lem2_rho^kappa_to_rho^infty_L^2}Let $\rho_{0}:[0,1]\to[0,1]$ be a measurable function, such that $\rho_0 - \bar\rho^{\infty} \in \mc H^{\gamma/2}$. Furthermore, let  $\rho^{\kappa}$ and $\rho^{\infty}$ {be} the  weak solutions of (\ref{eq:Dirichlet Equation}) and (\ref{eq:Dirichlet Equation_infty}), respectively, and with the same initial condition $\rho_{0}$. {Let $\hat\rho^{\kappa}_{t} := \rho^{\kappa}_{t/\kappa} $ , for all $t\in [0,T]$}. Then  $\hat\rho^{\kappa}$ converges strongly to $\rho^{\infty}$ in $L^{2}(0,T;L^{2})$, as $\kappa$ goes to $\infty$.
\begin{proof}
It is enough to show that 
\begin{equation}\label{Imp_Inf} 
\int_0^t  \Vert \hat\rho^{\kappa}_{s}-\rho^{\infty}_{s} \Vert^{2} \,ds =\int_0^t  \Vert \hat\vf_{s}^{\kappa} -\vf_{s}^{\infty} \Vert^{2} \,ds  \ls \dfrac{1}{\sqrt{\kappa}},
\end{equation}
for all $t\in [0,T]$  where  $\hat \varphi ^{\kappa}_{t}= \hat\rho_{t}^{\kappa} -\bar\rho^{\infty}$ and $\vf^{\infty}_{t} = (\rho_{0}-\bar\rho^{\infty})e^{-tV_{1}}$. {It is not difficult to see that $\hat \varphi ^{\kappa}_{t}$ satisfies
\begin{equation}\label{hatvf^k}
\begin{split}
&\langle \hat\vf_{t}^{\kappa} , G_{t} \rangle  - \langle \vf_{0}, G_{0} \rangle 
- \int_0^t\langle \hat\vf_{s}^{\kappa} , \partial_sG_{s} \rangle \, ds + \int_0^t\langle \hat\vf_{s}^{\kappa},G_s\rangle_{V_{1}} ds-\dfrac{1}{\kappa}\int_0^t\langle \hat\rho^{\kappa}_{s},\LL G_s\rangle ds =0
\end{split}   
\end{equation}
for all functions $G \in C^{1,\infty}_{c}([0,T]\times(0,1))$}. Then, calling  $\hat\delta^k:=\hat\vf^{\kappa} -\vf^{\infty}$ we have that
\begin{equation}\label{vf^k-vf^infty}
\begin{split}
&\langle \hat\delta_{t}^{\kappa}, G_{t} \rangle 
- \int_0^t\left\langle \hat\delta_{s}^{\kappa}, \left(\dfrac{1}{\kappa}\LL +\partial_s\right) G_{s} \right\rangle \, ds + \int_0^t\left\langle \hat\delta_{s}^{\kappa}, G_{s} \right\rangle_{V_{1}}  = \dfrac{1}{\kappa}\int_0^t\langle \rho_{s}^{\infty},G_s\rangle_{\gamma/2} ds
\end{split}   
\end{equation}
for any function $G \in C^{1,\infty}_{c}([0,T]\times(0,1))$. Let $\lbrace\hat H_{n}^{\kappa}\rbrace_{n\geq 1}$, be a sequence of functions in $C_{c}^{1,\infty}([0,T],(0,1))$ converging to $\hat \delta^{\kappa}$ with respect to the norm of $L^{2}(0,T;\mc H^{\gamma/2}_{0})$. Now, for $n\geq 1$ we define the test function $\hat G_{n}^{\kappa}(s,u) = \int^{t}_{s} \hat H_{n}^{\kappa}(r,u)dr $. Plugging $\hat G_{n}^{\kappa}$ into (\ref{vf^k-vf^infty}) and {using a similar argument as in proof of Lemma \ref{lem2_rho^kappa_to_rho^0_L^2}} we get that  

\begin{equation*}
\begin{split}
& 
 \int_{0}^t \Vert \hat\delta_{s}^{\kappa} \Vert^{2} \, ds  + \dfrac{1}{2\kappa}  \left \Vert \int_{0}^{t} \hat\delta_{s}^{\kappa}ds \right\Vert_{\gamma/2}^{2} +  \dfrac{1}{2}  \left \Vert \int_{0}^{t} \hat\delta_{s}^{\kappa}ds \right\Vert_{V_{1}}^{2} = \dfrac{1}{\kappa}\int_{0}^t\left\langle \rho^{\infty}_{s}, \int_{s}^{t} \hat\delta_{r}^{\kappa} dr\right\rangle_{\gamma/2} ds.
\end{split}   
\end{equation*}
By neglecting terms we get that 
\begin{equation*}
   \int_{0}^t \Vert \hat\rho_{s}^{\kappa}-\rho^{\infty}_{s} \Vert^{2} \, ds =\int_{0}^t \Vert \hat\delta_{s}^{\kappa} \Vert^{2} \, ds \leq \dfrac{1}{\kappa}\int_{0}^t\left\langle \rho^{\infty}_{s}, \int_{s}^{t} \hat\delta_{r}^{\kappa} dr\right\rangle_{\gamma/2} ds.
\end{equation*}
Then it is suffices to show that
\begin{equation*}
\dfrac{1}{\kappa}\int_{0}^t\left\langle \rho^{\infty}_{s}, \int_{s}^{t} \hat\delta_{r}^{\kappa} dr\right\rangle_{\gamma/2} ds\ls\dfrac{1}{\sqrt{\kappa}}
\end{equation*}
Indeed, by using twice the Cauchy-Schwarz's inequality we have that the term at the left hand side of the previous expression 
is bounded from above by
\begin{equation*} 
\dfrac{1}{\kappa}\int_{0}^t \Vert \rho_{s}^{\infty}\Vert_{\gamma/2} \left\Vert \int_{s}^{t} \hat\delta_{r}^{\kappa} dr\right\Vert_{\gamma/2} ds  \leq \dfrac{1}{\kappa} \sqrt{\int_{0}^t \Vert \rho_{s}^{\infty}\Vert^{2}_{\gamma/2} ds}\sqrt{\int_{0}^t\left\Vert \int_{s}^{t} \hat\delta_{r}^{\kappa} dr\right\Vert_{\gamma/2}^{2} ds}.  
\end{equation*}
 Since {by hypothesis} $\rho_{0}-\bar\rho^{\infty} \in \mc H^{\gamma/2}$ we know from item $i)$ of Lemma \ref{lem1_existence_rho^infty} that $\rho^{\infty}\in L^{2}(0,T;\mc H^{\gamma/2})$. Thus, from the latter and by the triangular inequality, the right hand side in the previous expression can be bounded from above {by a constant times}
\begin{equation*} 
\dfrac{1}{\kappa} \sqrt{\int_{0}^t\left( \int_{s}^{t} \Vert\hat\delta_{r}^{\kappa}\Vert_{\gamma/2} dr\right)^{2} ds} \,\ls\, \dfrac{1}{\kappa} \sqrt{t\left( \int_{0}^{t} \Vert\hat\delta_{r}^{\kappa}\Vert_{\gamma/2} dr\right)^{2}}.
\end{equation*}
By using again the Cauchy-Schwarz's inequality, the term on the right hand side in the last expression is bounded from above by
\begin{equation*}
\begin{split}
\dfrac{1}{\kappa} \sqrt{t^{2} \int_{0}^{t} \Vert\hat\delta_{r}^{\kappa}\Vert_{\gamma/2}^{2} dr}
&= \dfrac{1}{\kappa} \sqrt{t^{2} \int_{0}^{t} \Vert\hat\rho_{r}^{\kappa}-\rho^{\infty}_{r}\Vert_{\gamma/2}^{2} dr}\\
& \ls\, \dfrac{1}{\kappa} \sqrt{2t^{2} \int_{0}^{t} \Vert\hat\rho_{r}^{\kappa}\Vert_{\gamma/2}^{2}+\Vert\rho^{\infty}_{r}\Vert_{\gamma/2}^{2} dr}.
\end{split}
\end{equation*}
In the last inequality we used the Minkowski's inequality and the fact that $(a+b)^{2} \leq 2a^{2}+2b^{2}$. Now, since $ \int_{0}^{t} \Vert\hat\rho_{r}^{\kappa}\Vert_{\gamma/2}^{2}dr \ls\kappa$ {(this is due to item i) of Theorem \ref{Energy_Thm1} and a change of variables)} and $\rho^{\infty}\in L^{2}(0,T;\mc H^{\gamma/2})$  we can see that 
\begin{equation*}
\dfrac{1}{\kappa} \sqrt{2t^{2} \int_{0}^{t} \Vert\hat\rho_{r}^{\kappa}\Vert_{\gamma/2}^{2}+\Vert\rho^{\infty}_{r}\Vert_{\gamma/2}^{2} dr}\ls \dfrac{1}{\kappa} \sqrt{ \kappa + 1 } \ls \dfrac{1}{\sqrt{\kappa}},
\end{equation*}
 and we are done.
\end{proof}
\end{lem}

\subsection{Proof of item i) of Theorem \ref{convergence_rho^k_to_rho^0}.}
Recall $\vf_{t}^{\kappa}$ defined in \eqref{dvf}. Note that it is enough to show \eqref{Imp} with $\Vert \cdot \Vert$ replaced with $\Vert \cdot \Vert_{\gamma/2}$. 
{From (\ref{vf^k-vf^0}) we obtain, for $\ve >0$,} that
\begin{equation}\label{Diff}
\begin{split}
&\langle \delta_{t+\ve}^{\kappa} , G_{t+\ve} \rangle -\langle \delta_{t}^{\kappa}, G_{t} \rangle 
- \int_t^{t+\ve}\langle \delta_{s}^{\kappa} ,\left(\LL  + \partial_s\right) G_{s} \rangle \, ds  =- {\kappa}\int_t^{t+\ve}\langle \vf_{s}^{\kappa},G_s\rangle_{V_{1}} ds
\end{split}   
\end{equation}
for any function $G\in C_{c}^{1,\infty}([0,T]\times[0,1])$. Let $\lbrace H_{n}^{\kappa}\rbrace_{n\geq 1}$ be a sequence of functions in $C_{c}^{1,\infty}([0,T],(0,1))$ converging to $\delta^{\kappa}$ with respect to the norm of $L^{2}(0,T;\mc H^{\gamma/2}_{0})$ as $n\to \infty$. Now, for $n\geq 1$, we define the test function $G_{n}^{\kappa}(u) = \tfrac{1}{\ve}\int_{t}^{t+\ve} H_{n}^{\kappa}(r,u)dr $. Plugging $G_{n}^{\kappa}$ into last equality and taking $n\to\infty$, {a similar argument to the one of the proof of  Lemma \ref{lem2_rho^kappa_to_rho^0_L^2}} allows to get  
\begin{equation*}
\begin{split}
&\frac{1}{\ve}\left\langle \delta_{t+\ve}^{\kappa} -\delta_{t}^{\kappa}, \int_{t} ^{t+\ve} \delta_{r}^{\kappa} dr \right\rangle +\ve  \left \Vert\frac{1}{\ve} \int_{t}^{t+\ve}\delta_{r}^{\kappa} dr \right\Vert^{2}_{\gamma/2} 
= {\kappa}\int_t^{t+\ve}\left\langle \vf_{s}^{\kappa},\frac{1}{\ve}\int_{t} ^{t+\ve} \delta_{r}^{\kappa} dr\right\rangle_{V_{1}} ds.
\end{split}   
\end{equation*}
Integrating {last equality} over $[0,\tilde{t}]$ we get:
\begin{multline}\label{Fvf1}
\ve \int_0^{\tilde{t}} \left \Vert\frac{1}{\ve} \int_{t}^{t+\ve}\delta_{r}^{\kappa} dr \right\Vert^{2}_{\gamma/2} dt = {\kappa}\int_0^{\tilde{t}}\int_t^{t+\ve}\left\langle \vf_{s}^{\kappa},\frac{1}{\ve}\int_{t} ^{t+\ve} \delta_{r}^{\kappa}dr\right\rangle_{V_{1}} ds \, dt\\-\frac{1}{\ve}\int_0^{\tilde{t}}\left\langle \delta_{t+\ve}^{\kappa} -\delta_{t}^{\kappa}, \int_{t} ^{t+\ve} \delta_{r}^{\kappa} dr \right\rangle  dt . 
\end{multline}   
Now we use the Cauchy-Schwarz's inequality, Hardy's inequality and  (\ref{EE}) {to get that}  
\begin{equation}\label{B4}
\begin{split}
&  {\kappa}\int_0^{\tilde{t}}\int_t^{t+\ve}\left\langle \vf_{s}^{\kappa},\frac{1}{\ve}\int_{t} ^{t+\ve} \delta_{r}^{\kappa}dr\right\rangle_{V_{1}} ds \, dt\ls {\kappa} \int_0^{\tilde{t}}\int_t^{t+\ve}\Vert \vf_{s}^{\kappa}\Vert_{\gamma/2} \left\Vert \frac{1}{\ve}\int_{t} ^{t+\ve} \delta_{r}^{\kappa}dr\right\Vert_{\gamma/2} ds \, dt\\
\ls & {\kappa} \sqrt{\int_0^{\tilde{t}}\int_t^{t+\ve}\Vert \vf_{s}^{\kappa}\Vert_{\gamma/2}^{2} ds dt} \sqrt{ \int_0^{\tilde{t}}\int_t^{t+\ve}\left\Vert \frac{1}{\ve}\int_{t} ^{t+\ve} \delta_{r}^{\kappa}dr\right\Vert_{\gamma/2}^{2} ds \, dt }\\
\ls &\kappa \ve \sqrt {{\tilde{t}}}\sqrt {\int_0^{\tilde{t}} \left \Vert\frac{1}{\ve} \int_{t}^{t+\ve}\delta_{r}^{\kappa} dr \right\Vert^{2}_{\gamma/2} dt}.
\end{split}   
\end{equation}
Let us estimate the second term on the right hand side  (\ref{Fvf1}). First note that by  changing variables we have that
\begin{equation}\label{B01}
\begin{split}
-\frac{1}{\ve}\int_0^{\tilde{t}}\left\langle \delta_{t+\ve}^{\kappa} -\delta_{t}^{\kappa}, \int_{t} ^{t+\ve} \delta_{r}^{\kappa} dr \right\rangle  dt& =  \frac{1}{\ve}\int_0^{\tilde{t}}\int_{t} ^{t+\ve}\langle \delta_{t}^{\kappa},  \delta_{r}^{\kappa}  \rangle dr dt-\frac{1}{\ve}\int_0^{\tilde{t}}\int_{t} ^{t+\ve} \langle \delta_{t+\ve}^{\kappa}, \delta_{r}^{\kappa} \rangle  dr dt\\
&= \frac{1}{\ve}\int_0^{\tilde{t}}\int_{r} ^{r+\ve}\langle \delta_{t}^{\kappa},  \delta_{r}^{\kappa}  \rangle dt dr-\frac{1}{\ve}\int_\ve^{\tilde{t}+\ve}\int_{t -\ve} ^{t} \langle \delta_{t}^{\kappa}, \delta_{r}^{\kappa} \rangle  dr dt \\
\end{split}
\end{equation}
The term $\frac{1}{\ve}\int_0^{\tilde{t}}\int_{r} ^{r+\ve}\langle \delta_{t}^{\kappa},  \delta_{r}^{\kappa}  \rangle dt dr$ can be split as 
\begin{equation*}
\frac{1}{\ve} \left( \int_0^{\ve}\int_{r} ^{\ve} \langle \delta_{t}^{\kappa}, \delta_{r}^{\kappa} \rangle  dt dr + \int_0^{\ve}\int_{\ve}^{r+\ve} \langle \delta_{t}^{\kappa}, \delta_{r}^{\kappa} \rangle  dt dr + \int_{\ve}^{\tilde{t}}\int_{r}^{r+\ve}\langle \delta_{t}^{\kappa}, \delta_{r}^{\kappa} \rangle  dt dr\right).
\end{equation*}
{By Fubini's theorem}, we have that the term $\frac{1}{\ve} \int_\ve^{\tilde{t}+\ve}\int_{t -\ve} ^{t} \langle \delta_{t}^{\kappa}, \delta_{r}^{\kappa} \rangle  dr dt$ which appears in (\ref{B01}) is equal to
\begin{equation*}
\frac{1}{\ve} \left( \int_0^{\ve}\int_{\ve} ^{r+\ve} \langle \delta_{t}^{\kappa}, \delta_{r}^{\kappa} \rangle  dt dr + \int_\ve^{\tilde{t}}\int_{r}^{r+\ve} \langle \delta_{t}^{\kappa}, \delta_{r}^{\kappa} \rangle  dt dr + \int_{\tilde{t}}^{\tilde{t}+\ve}\int_{r}^{\tilde{t}+\ve}\langle \delta_{t}^{\kappa}, \delta_{r}^{\kappa} \rangle  dt dr\right).
\end{equation*}

Therefore we can write the second term on the right hand side of  (\ref{Fvf1}) as
\begin{equation}\label{B4_1}
\begin{split}
&-\frac{1}{\ve}\int_{{\tilde{t}}}^{{\tilde{t}}+\ve}\int_r^{{\tilde{t}}+\ve}\langle \delta_{t}^{\kappa},\delta_{r}^{\kappa} \rangle dt\, dr +\frac{1}{\ve}\int_0^{\ve}\int_r^{\ve}\langle \delta_{t}^{\kappa},\delta_{r}^{\kappa}\rangle dt\, dr\\
&\leq \frac{1}{\ve}\int_{{\tilde{t}}}^{{\tilde{t}}+\ve}\int_{\tilde{t}}^{{\tilde{t}}+\ve} \| \delta_{t}^{\kappa}\| \|\delta_{r}^{\kappa}\| dt\, dr + \frac{1}{\ve}\int_0^{\ve}\int_0^{\ve}\| \delta_{t}^{\kappa}\| \|\delta_{r}^{\kappa}\| dt\, dr\\
& =\frac{1}{\ve}\left(\int_{\tilde{t}}^{{\tilde{t}}+\ve}  \|\delta_{t}^{\kappa}\|\, dt \right)^{2}+ \frac{1}{\ve}\left(\int_0^{\ve} \|\delta_{t}^{\kappa}\|\, dt\right)^{2}\\ 
& \leq \int_{{\tilde{t}}}^{{\tilde{t}}+\ve}\| \delta_{t}^{\kappa}\|^{2} dt +\int_0^{\ve}\| \delta_{t}^{\kappa}\|^{2}  dt.\\
\end{split}
\end{equation}
{where in the inequalities above we used} the Cauchy-Schwarz's inequality. Then, using (\ref{B4}) and (\ref{B4_1}) in (\ref{Fvf1}) we obtain that 
\begin{equation}
\begin{split}
\int_0^{\tilde{t}} \left \Vert\frac{1}{\ve} \int_{t}^{t+\ve}\delta_{r}^{\kappa} dr \right\Vert^{2}_{\gamma/2} dt &\ls  \kappa \sqrt {{\tilde{t}}}\sqrt {\int_0^{\tilde{t}} \left \Vert\frac{1}{\ve} \int_{t}^{t+\ve}\delta_{r}^{\kappa} dr \right\Vert^{2}_{\gamma/2} dt} \\
&+ \dfrac{1}{\ve}\int_{{\tilde{t}}}^{{\tilde{t}}+\ve}\| \delta_{t}^{\kappa}\|^{2} dt +\dfrac{1}{\ve}\int_0^{\ve}\| \delta_{t}^{\kappa}\|^{2}  dt.
\end{split}   
\end{equation}
Taking $\ve\to 0 $, using  Lebesgue's differentiation theorem (see Theorem 1.35 in  \cite{Roub})  and the fact that $\delta_{0}^{\kappa} = 0$ {(since the initial condition for $\rho^\kappa$ and $\rho^0$ is the same)} we get that 
\begin{equation*}
\begin{split}
\int_0^{\tilde{t}} \Vert\delta_{t}^{\kappa}  \Vert^{2}_{\gamma/2} dt &\ls \kappa \sqrt {{\tilde{t}}}\sqrt {\int_0^{\tilde{t}}  \Vert \delta_{t}^{\kappa}\Vert^{2}_{\gamma/2} dt} +\| \delta_{\tilde{t}}^{\kappa}\|^{2},
\end{split}   
\end{equation*}
for all $\tilde{t} \in [0,T]$. Integrating {last inequality} over $[0,T]$ and using the Cauchy-Schwarz's inequality and using (\ref{Imp}) we conclude that
\begin{equation}
\begin{split}\label{B6}
\int^{T}_{0}\int_0^{\tilde{t}} \Vert\delta_{t}^{\kappa}  \Vert^{2}_{\gamma/2} dtd\tilde{t} 
&\ls \kappa T \sqrt {\int_{0}^{T}\int_0^{\tilde{t}}  \Vert \delta_{t}^{\kappa}\Vert^{2}_{\gamma/2} dtd\tilde{t}} + \kappa T^{2},
\end{split}   
\end{equation}
in the last inequality we have used (\ref{Imp}). Then, by a simple computation  we have that
\begin{equation}\label{B6.1}
\begin{split}
\int_{0}^{T}\int_0^{\tilde{t}} \Vert\delta_{t}^{\kappa}  \Vert^{2}_{\gamma/2} dtd\tilde{t} &\ls \kappa T^{2}.
\end{split}   
\end{equation}
{By Fubini's theorem}, we get that 
\begin{equation}\label{B7}
\int_0^{T} \int_{0}^{\tilde{t}} \Vert \delta_{t}^{\kappa} \Vert^{2}_{\gamma/2} \, dt d\tilde{t}=\int_0^T(T-t) \Vert \delta_{t}^{\kappa} \Vert^{2}_{\gamma/2} \, dt \geq\frac{T}{2}\int_0^{T/2} \Vert \delta_{t}^{\kappa}\Vert^{2}_{\gamma/2} \, dt.
\end{equation}
The result now follows from (\ref{B6.1}) and (\ref{B7}). 
\begin{flushright}
$\qed$
\end{flushright}

\subsection{Proof of item ii) of Theorem \ref{convergence_rho^k_to_rho^0}} 
Recall $\hat\vf_{t}^{\kappa}$ and $\vf_{t}^{\infty}$ defined in Lemma \ref{lem2_rho^kappa_to_rho^infty_L^2}. It is enough to show (\ref{Imp_Inf}) with $\Vert \cdot\Vert$ replaced with $\Vert \cdot\Vert_{V_{1}}$:
\begin{equation}\label{Imp2_infty}
\int_0^T  \Vert \hat\vf_{t}^{\kappa} -\vf_{t}^{\infty} \Vert^{2}_{V_{1}} \,dt \ls \dfrac{1}{\sqrt{\kappa}}.
\end{equation}
{From (\ref{vf^k-vf^infty}), we obtain, for $\ve >0$, that}
\begin{multline}\label{Diff_infty}
\langle \hat\delta_{t+\ve}^{\kappa} , G_{t+\ve} \rangle -\langle \hat\delta_{t}^{\kappa}, G_{t} \rangle 
- \int_t^{t+\ve}\langle \hat\delta_{s}^{\kappa} ,\Big(\dfrac{1}{\kappa}\LL  + \partial_s\Big) G_{s} \rangle \, ds\\ +\int_t^{t+\ve}\langle \hat\delta_{s}^{\kappa} , G_{s} \rangle_{V_{1}} \, ds  = \dfrac{1}{\kappa}\int_t^{t+\ve}\langle \rho_{s}^{\infty},G_s\rangle_{\gamma/2} ds  
\end{multline}
for any function $G\in C_{c}^{1,\infty}([0,T]\times[0,1])$. Let $\lbrace\hat H_{n}^{\kappa}\rbrace_{n\geq 1}$ be a sequence of functions in $C_{c}^{1,\infty}([0,T],(0,1))$ converging to $\hat\delta^{\kappa}$ with respect to the norm of $L^{2}(0,T;\mc H^{\gamma/2}_{0})$ as $n\to\infty$. Now, for $n\geq 1$ we define the test functions $\hat G_{n}^{\kappa}(u) = \tfrac{1}{\ve}\int_{t}^{t+\ve} \hat H_{n}^{\kappa}(r,u)dr $. Plugging $\hat G_{n}^{\kappa}$ into (\ref{Diff_infty}) and taking $n\to\infty$, a similar argument to the one of the {proof of  Lemma \ref{lem2_rho^kappa_to_rho^0_L^2} allows to get}  
\begin{multline} \label{Fvf_infty}
\frac{1}{\ve}\left\langle \hat\delta_{t+\ve}^{\kappa} -\hat\delta_{t}^{\kappa}, \int_{t} ^{t+\ve} \hat\delta_{r}^{\kappa} dr \right\rangle +\dfrac{\ve}{\kappa}  \left \Vert\frac{1}{\ve} \int_{t}^{t+\ve}\hat\delta_{r}^{\kappa} dr \right\Vert^{2}_{\gamma/2} \\ +\ve \left \Vert\frac{1}{\ve} \int_{t}^{t+\ve}\hat\delta_{r}^{\kappa} dr \right\Vert^{2}_{V_{1}}
= \dfrac{1}{\kappa}\int_t^{t+\ve}\left\langle \rho_{s}^{\infty},\frac{1}{\ve}\int_{t} ^{t+\ve} \hat\delta_{r}^{\kappa} dr\right\rangle_{\gamma/2} ds.   
\end{multline}
By neglecting the term $\dfrac{\ve}{\kappa}  \left \Vert\frac{1}{\ve} \int_{t}^{t+\ve}\hat\delta_{r}^{\kappa} dr \right\Vert^{2}_{\gamma/2} $ in (\ref{Fvf_infty}) and then integrating over $[0,\tilde{t}]$   we get that
\begin{multline}\label{Fvf1_infty}
 \ve\int_0^{\tilde{t}} \left \Vert\frac{1}{\ve} \int_{t}^{t+\ve}\hat\delta_{r}^{\kappa} dr \right\Vert^{2}_{V_{1}} dt \leq \dfrac{1}{\kappa}\int_0^{\tilde{t}}\int_t^{t+\ve}\left\langle \rho_{s}^{\infty},\frac{1}{\ve}\int_{t} ^{t+\ve} \hat\delta_{r}^{\kappa}dr\right\rangle_{\gamma/2} ds \, dt\\-\frac{1}{\ve}\int_0^{\tilde{t}}\left\langle \hat\delta_{t+\ve}^{\kappa} -\hat\delta_{t}^{\kappa}, \int_{t} ^{t+\ve} \hat\delta_{r}^{\kappa} dr \right\rangle  dt. 
\end{multline}   
Now we use twice the Cauchy-Schwarz's inequality in order to get that the first term on the right hand side in the previous expression is bounded from above by 
\begin{equation}\label{B4_infty}
\begin{split}
& \dfrac{1}{\kappa} \int_0^{\tilde{t}}\int_t^{t+\ve}\Vert \rho_{s}^{\infty}\Vert_{\gamma/2} \left\Vert \frac{1}{\ve}\int_{t} ^{t+\ve} \hat\delta_{r}^{\kappa}dr\right\Vert_{\gamma/2} ds \, dt\\
&\leq \dfrac{1}{\kappa} \sqrt{\int_0^{\tilde{t}}\int_t^{t+\ve}\Vert \rho_{s}^{\infty}\Vert_{\gamma/2}^{2} ds dt} \sqrt{ \int_0^{\tilde{t}}\int_t^{t+\ve}\left\Vert \frac{1}{\ve}\int_{t} ^{t+\ve} \hat\delta_{r}^{\kappa}dr\right\Vert_{\gamma/2}^{2} ds \, dt }\\
&\leq \dfrac{\sqrt{\ve}}{\kappa}\sqrt{\int_0^{\tilde{t}}\int_t^{t+\ve}\Vert \rho_{s}^{\infty}\Vert_{\gamma/2}^{2} ds dt}\sqrt {\int_0^{\tilde{t}} \left \Vert\frac{1}{\ve} \int_{t}^{t+\ve}\hat\delta_{r}^{\kappa} dr \right\Vert^{2}_{\gamma/2} dt}.
\end{split}   
\end{equation}
By a similar argument  {as the one} in the proof of item i) of Theorem \ref{convergence_rho^k_to_rho^0}  we have  that the second term on the right hand side in (\ref{Fvf1_infty}) is bounded from above by 
\begin{equation}\label{B41_infty}
\dfrac{1}{\ve}\int_{{\tilde{t}}}^{{\tilde{t}}+\ve}\| \hat\delta_{t}^{\kappa}\|^{2} dt +\dfrac{1}{\ve}\int_0^{\ve}\| \hat\delta_{t}^{\kappa}\|^{2}  dt.
\end{equation}
Therefore, by using (\ref{B4_infty}) and (\ref{B41_infty}) in (\ref{Fvf1_infty}) we get that  
\begin{equation}
\begin{split}
\int_0^{\tilde{t}} \left \Vert\frac{1}{\ve} \int_{t}^{t+\ve}\hat\delta_{r}^{\kappa} dr \right\Vert^{2}_{V_{1}} dt &\leq  \dfrac{1}{\kappa}\sqrt{\int_0^{\tilde{t}}\dfrac{1}{\ve}\int_t^{t+\ve}\Vert \rho_{s}^{\infty}\Vert_{\gamma/2}^{2} ds dt}\sqrt {\int_0^{\tilde{t}} \left \Vert\frac{1}{\ve} \int_{t}^{t+\ve}\hat\delta_{r}^{\kappa} dr \right\Vert^{2}_{\gamma/2} dt} \\
&+ \dfrac{1}{\ve}\int_{{\tilde{t}}}^{{\tilde{t}}+\ve}\| \hat\delta_{t}^{\kappa}\|^{2} dt +\dfrac{1}{\ve}\int_0^{\ve}\| \hat\delta_{t}^{\kappa}\|^{2}  dt.
\end{split}   
\end{equation}
Taking $\ve\to 0 $, using  Lebesgue's differentiation theorem (see Theorem 1.35 in  \cite{Roub}) and the fact that $\hat\delta_{0}^{\kappa} = 0$ we get that 
\begin{equation*}
\begin{split}
\int_0^{\tilde{t}} \Vert\hat\delta_{t}^{\kappa}  \Vert^{2}_{V_{1}} dt &\leq  \dfrac{1}{\kappa}\sqrt{\int_0^{\tilde{t}} \Vert \rho_{t}^{\infty}\Vert_{\gamma/2}^{2}  dt}\sqrt {\int_0^{\tilde{t}}  \Vert \hat\delta_{t}^{\kappa}\Vert^{2}_{\gamma/2} dt} +\| \hat\delta_{\tilde{t}}^{\kappa}\|^{2},
\end{split}   
\end{equation*}
for all $\tilde{t} \in [0,T]$. Integrating the {previous expression} over $[0,T]$ and using the Cauchy-Schwarz's inequality we get that
\begin{equation}
\begin{split}\label{B6_infty}
\int^{T}_{0}\int_0^{\tilde{t}} \Vert\hat\delta_{t}^{\kappa}  \Vert^{2}_{V_{1}} dtd\tilde{t} &\leq \dfrac{1}{\kappa}\sqrt {\int_{0}^{T}\int_0^{\tilde{t}}  \Vert \rho_{t}^{\infty}\Vert^{2}_{\gamma/2} dtd\tilde{t}} \sqrt {\int_{0}^{T}\int_0^{\tilde{t}}  \Vert \hat\delta_{t}^{\kappa}\Vert^{2}_{\gamma/2} dtd\tilde{t}} + \int_{0}^{T}\| \hat\delta_{\tilde{t}}^{\kappa}\|^{2}d\tilde{t}\\
&\ls  \dfrac{1}{\kappa} \sqrt {\int_{0}^{T}\int_0^{T}  \Vert \hat\delta_{t}^{\kappa}\Vert^{2}_{\gamma/2} dtd\tilde{t}} +\dfrac{1}{\sqrt{\kappa}},\\
&\ls  \dfrac{1}{\kappa} \sqrt {2T \int_0^{T}  \Vert \hat\rho_{t}^{\kappa}\Vert^{2}_{\gamma/2}+ \Vert \rho_{t}^{\infty}\Vert^{2}_{\gamma/2} dt} +\dfrac{1}{\sqrt{\kappa}},\\
&\ls  \dfrac{1}{\kappa} \sqrt {(\kappa +2) } +\dfrac{1}{\sqrt{\kappa}}.
\end{split}   
\end{equation}
In the second inequality {above} we used the fact that $\rho^{\infty}\in L^{2}(0,T;\mc H^{\gamma/2})$ (see item $i)$ of Lemma \ref{lem1_existence_rho^infty}) and  (\ref{Imp2_infty}), {while} in the third inequality of we used Minkoski's inequality and the fact that $(a+b)^{2}\leq 2a^{2} +2b^{2}$. And finally, the last 
inequality of (\ref{B6_infty}) is true since $\rho^{\infty}\in L^{2}(0,T;\mc H^{\gamma/2})$ and item i) of Theorem \ref{Energy_Thm1}.

Then, by a simple computation  we have that
\begin{equation}\label{B6.1_infty}
\begin{split}
\int_{0}^{T}\int_0^{\tilde{t}} \Vert\hat\delta_{t}^{\kappa}  \Vert^{2}_{V_{1}} dtd\tilde{t} &\ls \dfrac{ 1}{\sqrt{\kappa}}.
\end{split}   
\end{equation}
By Fubini's theorem, we have that  
\begin{equation}\label{B7_infty}
\int_0^{T} \int_{0}^{\tilde{t}} \Vert \hat\delta_{t}^{\kappa} \Vert^{2}_{V_{1}} \, dt d\tilde{t} =\int_0^T(T-t) \Vert \hat\delta_{t}^{\kappa} \Vert_{V_{1}}^{2} \, dt \geq \frac{T}{2}\int_0^{T/2} \Vert \hat\delta_{t}^{\kappa}\Vert^{2}_{V_{1}} \, dt. 
\end{equation}
The result now follows from (\ref{B6.1_infty}) and (\ref{B7_infty}). 

\begin{flushright}
$\qed$
\end{flushright}
\section{Proof of Theorem \ref{Existence_uniqueness_convergence_stationary}} 
\label{sec:Exis_unq_sta_sol}
In this section  we prove items i) and ii) of  Theorem \ref{Existence_uniqueness_convergence_stationary}.
Now we are interested in analyzing the convergence of the stationary solution $\bar\rho^\kappa$ as $\kappa\to 0$ and $\kappa\to\infty$. From Definition \ref{Def. Stationary_RFRD}, for $\kappa \geq 0,$ and  for $\bar\vf^{\kappa} = \bar\rho^{\kappa}-\bar\rho^{\infty}$ we have that $\bar\vf^{\kappa} \in \mc H^{\gamma/2}_{0}$ and 
\begin{equation}\label{weak_sta_eq}
\langle \bar\vf^{\kappa} ,-\LL G\rangle + \kappa \langle \bar\vf^{\kappa},G \rangle_{V_{1}} = I_{\bar\rho^{\infty}}(G),  
\end{equation}
 for any test function $G$ of compact support included in $(0,1)$. Above $I_{\bar\rho^{\infty}}:\mathcal{ H}^{\gamma/2}_{0}\to \RR$ is a linear form defined by $I_{\bar\rho^{\infty}}(G)=\langle \bar\rho^{\infty},\LL G \rangle$. Moreover, this linear form is continuous. Indeed, using integration by parts given in Proposition 3.3 in \cite{GM} we have that 
 \begin{equation}\label{ST1}
 \begin{split}
\vert I_{\bar\rho^{\infty}}(G)\vert &=  \left\vert \int_{0}^{1} \bar \rho^{\infty} (u) \LL G(u) du \right\vert \\
&=\dfrac{c_{\gamma}}{2}\left\vert  \iint_{[0,1]^{2}} \dfrac{( \bar\rho^{\infty} (u) - \bar\rho^{\infty} (v) )(G(u)-G(v))}{\vert u-v\vert^{\gamma+1}} dv du \right\vert  \\
&\leq \Vert  \bar\rho^{\infty}\Vert_{\gamma/2} \Vert  G \Vert_{\gamma/2} <\infty.
 \end{split}
 \end{equation}
Above we used the Cauchy-Schwarz's inequality and the fact that $\Vert  \bar\rho^{\infty}\Vert_{\gamma/2}$ is finite (see (\ref{bar_rho^infty_in_H})). Therefore, $ \vert I_{\rho^{\infty}}(G)\vert \ls \Vert G \Vert_{\mathcal{ H}^{\gamma/2}_{0}}.$ 

Then it is enough to analyze the behavior of $\bar\vf^{\kappa}$. We claim that we can take $G=\bar\vf^\kappa$ in (\ref{weak_sta_eq}). The justification is postponed to the end of the proof. Whence, from \eqref{ST1} we have that 
\begin{equation}\label{prop_u_kappa}
\|\bar\vf^\kappa\|_{\gamma/2}^2 + \kappa \|\bar\vf^\kappa\|_{V_{1}}^2 = I_{\bar\rho^{\infty}}(\bar\vf^\kappa)\ls\|\bar\vf^\kappa\|_{\gamma/2},
\end{equation}
from where we conclude that $ \|\bar\vf^\kappa\|_{\gamma/2} <\infty.$ Plugging this back into \eqref{prop_u_kappa} we get that
\begin{equation}\label{bound_u_kappa_V1}
 \|\bar\vf^\kappa\|_{V_{1}} \ls \frac{1}{\sqrt \kappa}. 
\end{equation}
Now, note that $\bar\vf^0 \in \mc H ^{\gamma/2}_{0} $ satisfies $\langle \bar\vf^{0} ,-\LL G\rangle = I_{\bar\rho^{\infty}}(G),$ for any function $G\in C^{\infty}_{c}((0,1))$.
Then $\bar\vf^\kappa-\bar\vf^0$ satisfies
\begin{equation*}
\langle \bar\vf^{\kappa}-\bar\vf^0 ,-\LL G\rangle + \kappa \langle \bar\vf^{\kappa},G \rangle_{V_{1}} = 0,  
\end{equation*}
for any function $G\in C^{\infty}_{c}((0,1))$. We claim that we can take $G=\bar\vf^\kappa-\bar\vf^0$ in the previous equality. The proof is analogous to the one done at the end of this section. Thus, we get that 
\begin{equation*}
\|\bar\vf^\kappa-\bar\vf^0\|_{\gamma/2}^2 = k\langle \bar\vf^\kappa, \bar\vf^0-\bar\vf^\kappa\rangle_{V_1}\leq {\kappa \|\bar\vf^\kappa\|_{V_1}\|\bar\vf^\kappa-\bar\vf^0\|_{V_1} }.
\end{equation*}
From \eqref{bound_u_kappa_V1} and fractional Hardy's inequality given in (\ref{norms_related}) we have that
\begin{equation*}
\|\bar\vf^\kappa-\bar\vf^0\|_{\gamma/2}^2 \ls {\sqrt{ \kappa} \|\bar\vf^\kappa-\bar\vf^0\|_{V_1} }\ls   \sqrt \kappa  \|\bar\vf^\kappa-\bar\vf^0\|_{\gamma/2},
\end{equation*}
from where we conclude that $ \|\bar\vf^\kappa-\bar\vf^0\|_{\gamma/2}\ls \sqrt \kappa.$  
  Then 
$\bar\vf^\kappa$ converges to $\bar\vf^0$, as $k\to 0$ in the $\|\cdot\|_{\gamma/2}$ norm. So far we proved item i).
\begin{rem}
From fractional Hardy's inequality (see \ref{norms_related}) the convergence is also true in $L^2_{V_1}$ and since  $$\|\bar\vf^\kappa-\bar\vf^0\|_{V_1}\geq  V_1(\tfrac{1}{2})\|\bar\vf^\kappa-\bar\vf^0\|$$ we conclude  that the convergence also holds in $L^2$. 
\end{rem}
For item ii),  by (\ref{bound_u_kappa_V1}) we get that $\Vert \bar\vf^{\kappa}\Vert_{V_{1}} \to 0$ and so $\Vert \bar\vf^{\kappa}\Vert\to 0$ as $k\to\infty$.

We conclude this proof by showing that we can take  $G = \bar\vf^{\kappa}$ in (\ref{weak_sta_eq}). Indeed, since $C_{c}^{\infty}((0,1))$ is dense in $\mathcal{H}^{\gamma/2}_{0}$, there exists a sequence $\lbrace\bar H_{n}^{\kappa}\rbrace_{n\ge 1}$ in $C_{c}^{\infty}((0,1))$ converging to $\bar\vf^{\kappa}$, i.e, $\Vert \bar H_{n}^{\kappa}-\bar\vf^{\kappa}\Vert_{\gamma/2} \to 0$ as $n\to \infty$. Observe that as a result of the latter and (\ref{norms_related}) we  also have $\Vert \bar H_{n}^{\kappa}-\bar\vf^{\kappa}\Vert_{V_{1}} \to 0$ as $n\to \infty$. Using the Cauchy-Schwarz's inequality we have that

\begin{equation*}
\begin{split} 
 \langle \bar\vf^{\kappa},\bar H^{\kappa}_{n} - \bar\vf^{\kappa}\rangle_{\gamma/2}& \leq \Vert \bar\vf^{\kappa}\Vert_{\gamma/2} \Vert \bar H^{\kappa}_{n} -\bar\vf^{\kappa}\Vert_{\gamma/2},\\  
\langle \bar\vf^{\kappa},\bar H^{\kappa}_{n} - \bar\vf^{\kappa}\rangle_{V_{1}} &\leq \Vert \bar\vf^{\kappa}\Vert_{V_{1}} \Vert \bar H^{\kappa}_{n} -\bar\vf^{\kappa}\Vert_{V_{1}},\\
 I_{\bar\rho^{\infty}}(\bar H^{\kappa }_{n}-\bar\vf^{\kappa}) &\leq \Vert \bar\rho^{\infty}\Vert_{\gamma/2} \Vert \bar H^{\kappa}_{n} -\bar\vf^{\kappa}\Vert_{\gamma/2}, 
\end{split}
\end{equation*} 
 all going to $0$ as $n\to\infty$.
Thus, we can rewrite (\ref{weak_sta_eq}) as 
\begin{eqnarray*}
\langle \bar\vf^{\kappa} ,-\LL \bar\vf^{\kappa}\rangle + \langle \bar\vf^{\kappa} ,-\LL ( \bar H^{\kappa}_{n} - \bar\vf^{\kappa})\rangle + \kappa ( \langle \bar\vf^{\kappa},\bar\vf^{\kappa} \rangle_{V_{1}} +\langle \bar\vf^{\kappa},\bar H^{\kappa}_{n} -\bar\vf^{\kappa} \rangle_{V_{1}} )= I_{\bar\rho^{\infty}}(\bar\vf^{\kappa}) + I_{\bar\rho^{\infty}}(\bar H^{\kappa}_{n}-\bar\vf^{\kappa}).  
\end{eqnarray*}
Now it is enough to take $n\to\infty$.
\begin{flushright}
$\qed$
\end{flushright}

\section{Uniqueness of weak solutions}
\label{sec:Uniqueness}
In this section we prove Lemmas \ref{lem:uniquess} and  \ref{existence and uniqueness}. For Lemma \ref{lem:uniquess}, we only focus  in the  proof of the uniqueness for the weak solutions of  (\ref{eq:Dirichlet Equation}) for $\hat\kappa =\kappa$. The proof of the uniqueness of the weak solutions of \eqref{eq:Dirichlet Equation} for $\kappa =0$ and \eqref{eq:Dirichlet Equation_infty} is analogous, the difference is that only the first two items in Lemma \ref{lem:unique-conv-007} below are required. Finally, in  Subsection \ref{subsec:EU} we prove Lemma \ref{existence and uniqueness}. 
\subsection{Proof of Lemma \ref{Existence_uniqueness_convergence_stationary}}\label{subsec:app-unique}
 Let $\rho^{\kappa,1}$ and $\rho^{\kappa,2}$ two weak solutions of (\ref{eq:Dirichlet Equation}) with the same initial condition and let us denote $\tilde\rho^{\kappa} = \rho^{\kappa,1} -\rho^{\kappa,2}$. For almost every $t\in [0,T]$, we identify ${ \tilde\rho^{\kappa}}_t$ with its continuous representation on $[0,1]$. Therefore, by Remark \ref{use:rem_dir} we have ${\tilde \rho}_t ^{\kappa}(0)={\tilde \rho}_t ^{\kappa} (1)=0$. Since ${\mc H}_0^{\gamma/2}$ is equal to the set of functions in ${\mc H}^{\gamma/2}$ vanishing at $0$ and $1$ we have that ${\tilde \rho}_t^{\kappa} \in {\mc H}_0^{\gamma/2}$ for a.e. time $t \in [0,T]$ and, in fact, ${\tilde \rho^{\kappa}} \in L^2 (0,T;{\mc H}_0^{\gamma/2})$. Moreover, for any $t \in [0,T]$ and all functions $G\in C_c^{1,\infty} ([0,T] \times (0,1))$ we have
\begin{equation}
\label{eq:unique007}
\langle {\tilde \rho^{\kappa}}_{t}, G_{t} \rangle - \int_0^t \left\langle {\tilde \rho}_{s}^{\kappa},\Big(\partial_s + \bb L \Big) G_{s} \right\rangle ds +\kappa \int^{t}_{0}\left\langle {\tilde \rho}_s ^{\kappa},  G_s  \right\rangle_{V_1} ds=0.
\end{equation}

Note that, it is easy to show that $C_c^{1,\infty} ([0,T] \times (0,1))$ is dense in $L^2 (0,T; {\mc H}_0^{\gamma/2})$. Let $\lbrace H_n^{\kappa}\rbrace_{n \ge 1}$ be a sequence of functions in  $C_c^{1,\infty} ([0,T] \times (0,1))$ converging to ${\tilde \rho^{\kappa}}$ with respect to the norm of $L^{2}(0,T;\mc{H}^{1/2}_{0})$ as $n\to \infty$. For $n\geq 1$, we define the test functions 
$\forall t \in [0,T], \quad \forall u \in [0,1], \quad G_n^{\kappa} (t,u)= \int_t^T  H_n^{\kappa} (s,u) \,ds.$ Plugging $G_n^{\kappa}$ into (\ref{eq:unique007}) and letting $n \to \infty$ we conclude by Lemma \ref{lem:unique-conv-007} below that
\begin{equation}
\int_0^T \Vert \tilde{\rho}^{\kappa}_{s}\Vert^{2} ds + \cfrac{1}{2} \; \Big\| \int_0^T {\tilde \rho}_s ^{\kappa}ds \Big\|^2_{\gamma/2} + \cfrac{\kappa}{2}\;  \Big\|\int_0^T {\tilde \rho}_s^{\kappa} ds  \, \Big\|_{V_{1}}^2 =0.
\end{equation}
Recall that  $\langle \cdot, \cdot \rangle_{V_{1}}$ (resp. $\| \cdot\|_{V_{1}}$) is the scalar product (resp. the norm) corresponding to the Hilbert space  $L^2_{V_{1}}$.

Then, it follows that for almost every time $s \in [0,T]$ the continuous function $\tilde \rho^{\kappa}_{s}$ is equal to $0$ and we conclude the uniqueness of the weak solutions to (\ref{eq:Dirichlet Equation}).

\begin{lem}
\label{lem:unique-conv-007}
Let $\lbrace G_n^{\kappa}\rbrace_{n\geq n}$ be defined as above. We have
\begin{enumerate}[i)]
\item $\displaystyle\lim_{n \to \infty}  \int_0^T \left\langle {\tilde \rho}_s^{\kappa}  , (\partial_s G_n^{\kappa}) (s,\cdot) \right\rangle ds = - \int_0^T  \Vert \tilde{\rho}^{\kappa}_{s}\Vert^{2} ds$. 
\item $\displaystyle\lim_{n \to \infty} \int_0^T \left\langle {\tilde \rho}_{s}^{\kappa},  \bb L G_{n}^{\kappa} (s,\cdot) \right\rangle ds =- \cfrac{1}{2} \; \Big\| \int_0^T {\tilde \rho}_s^{\kappa} ds \Big\|^2_{\gamma/2}.$ 
\item $\displaystyle\lim_{n \to \infty} \int^{T}_{0}\left\langle  {\tilde \rho}_s ^{\kappa},  G_n^{\kappa} (s,\cdot)  \right\rangle ds = \cfrac{1}{2}\;  \Big\|\int_0^T {\tilde \rho}_s ^{\kappa}ds  \, \Big\|_{V_{1}}^2 < \infty.$
\end{enumerate}
\end{lem}

\begin{proof} The proof of this lemma is quite similar to the proof of items i), ii) and iii) in the proof of Lemma \ref{lem2_rho^kappa_to_rho^0_L^2}. For that reason we just sketch the main steps of the proof and we leave the details to the reader. 
For i)  we have that
\begin{equation}
\begin{split}
&-\int_0^T \left\langle {\tilde \rho}_s ^{\kappa} , (\partial_s G_n^{\kappa}) (s,\cdot)\right\rangle ds  = \int_0^T  \big\langle {\tilde \rho}_s^{\kappa},H_n ^{\kappa}(s, \cdot) - {\tilde \rho}_s^{\kappa} \big\rangle ds
+ \int_0^T \| {\tilde \rho}_s ^{\kappa}\|^2  ds,
\end{split}
\end{equation}
and  by the Cauchy-Schwarz inequality,
\begin{equation}
\begin{split}
& \left| \int_0^T  \big\langle {\tilde \rho}_s^{\kappa} \, , \, H_n^{\kappa} (s, \cdot) - {\tilde \rho}_s ^{\kappa}\big\rangle \, ds \right| \le \sqrt{ \int_0^T  \| {\tilde \rho}_s^{\kappa} \|^2 \, ds} \; \sqrt{ \int_0^T  \| H_n ^{\kappa}(s, \cdot) - {\tilde \rho}_s^{\kappa} \|^2 \, ds } 
\end{split} 
\end{equation}
which goes to $0$ as $n \to \infty$.

For ii), we first use the integration by parts formula for the regional fractional Laplacian (see Theorem 3.3 in \cite{GM})  to get
$$ \int_0^T \left\langle {\tilde \rho}_{s}^{\kappa} , \bb L G_{n}^{\kappa} (s,\cdot) \right\rangle ds = - \int_0^T \Big\langle \tilde\rho_s^{\kappa} \,  , G_n^{\kappa} (s, \cdot) \, \Big\rangle_{\gamma/2}\, ds,$$
and as in ii) in the proof of Lemma \ref{lem2_rho^kappa_to_rho^0_L^2} we have that
\begin{equation*}
\begin{split}
&\int_0^T \Big\langle \tilde\rho_s ^{\kappa}\,  , G_n^{\kappa} (s, \cdot) \, \Big\rangle_{\gamma/2}\, ds =\cfrac{1}{2} \; \Big\| \int_0^T {\tilde \rho}_s ^{\kappa}ds \Big\|^2_{\gamma/2}+ \;  \int_0^T \Big\langle \tilde\rho_s^{\kappa} \,  , \int_s^T \left( H_n ^{\kappa}(t, \cdot) -{\tilde \rho}_t^{\kappa}\right)  dt \, \Big\rangle_{\gamma/2}\, ds.
\end{split}
\end{equation*}
Now, note that the term on the right hand side of last expression vanishes as $n\to\infty$ as a consequence of a successive use of Cauchy-Schwarz's inequalities. 
The proof of iii) is similar to the proof of ii) by using the fractional Hardy's inequality (see (\ref{norms_related})) 
and since $C_c^{\infty}((0,1))$ is dense in $H_0^{\gamma/2}$ we have that any $g\in H_0^{\gamma/2}$ is also in the space $L^2_{V_{1}}$ and that (\ref{norms_related}) remains valid for $g$. In particular, we have that the right hand side of iii) is finite.
We have
\begin{equation}
\begin{split}
&\int_0^T \Big\langle \tilde\rho_s^{\kappa} \,  , G_n^{\kappa} (s, \cdot) \, \Big\rangle_{V_{1}}\, ds =\cfrac{1}{2} \; \Big\| \int_0^T {\tilde \rho}_s^{\kappa} ds \Big\|^2_{V_{1}}+ \;  \int_0^T \Big\langle \tilde\rho_s^{\kappa} \,  , \int_s^T \left( H_n^{\kappa} (t, \cdot) -{\tilde \rho}_t^{\kappa}\right)  dt \, \Big\rangle_{V_{1}}\, ds.
\end{split}
\end{equation}
To conclude the proof of iii) it is sufficient to prove that \pat{ the term on the right hand side of last expression vanishes as $n\to\infty$}. But this is a consequence of a successive use of the Cauchy-Schwarz inequalities and Hardy's inequality, from which we get 
\begin{equation*}
\begin{split}
&\left| \int_0^T \Big\langle \tilde\rho_s^{\kappa} \,  , \int_s^T \left( H_n^{\kappa} (t, \cdot) -{\tilde \rho}_t^{\kappa}\right)  dt \, \Big\rangle_{V_{1}}\, ds \right| \\& \le C T  \, \sqrt{ \int_0^T \Big\| \tilde\rho_s ^{\kappa}\Big\|^2_{\gamma/2} ds} \; \sqrt{  \int_0^T \Big\|  H_n ^{\kappa}(t, \cdot) -{\tilde \rho}_t^{\kappa} \Big\|_{\gamma/2}^2\, dt} \; \xrightarrow[n \to \infty]{} \; 0.
\end{split}
\end{equation*}
The proof of the uniqueness of the weak solutions of \eqref{eq:Dirichlet Equation} for $\kappa =0$  is analogous, the difference is that  only the first two items in Lemma \ref{lem:unique-conv-007} above are required. The uniqueness of the weak solutions of  \eqref{eq:Dirichlet Equation_infty} is analogous as well, in this case  only  items i)  and iii)  in Lemma \ref{lem:unique-conv-007} above are required.

\end{proof}

\subsection{Proof of Lemma \ref{existence and uniqueness}}
\label{subsec:EU}

Recall \eqref{weak_sta_eq}. As we will see below, by Lax-Milgram's Theorem (see \cite{Brezis}), there exists a unique function $\bar\vf^{\hat\kappa} \in \mc H^{\gamma/2}_{0}$ which is solution of (\ref{weak_sta_eq}). Then, it is not difficult to see that $\bar\rho^{\hat\kappa}:= \bar\vf^{\hat\kappa} +\bar\rho^{\infty} $ is the desired weak solution of (\ref{eq:Stationary_RFRD}). For that purpose, let $a^{{\hat\kappa}}:\mathcal{ H}^{\gamma/2}_{0}\times\mathcal{ H}^{\gamma/2}_{0}\to \RR$ be the bilinear form defined, for $G,F\in\mathcal{ H}^{\gamma/2}_{0}$,  as 
\begin{equation}\label{bilinear_form_sta}
a^{{\hat\kappa}}(F,G)=\langle F , G \rangle_{\gamma/2} + \hat\kappa \langle F,G \rangle_{V_{1}}.
\end{equation} 
From Lax-Milgram Theorem, in order to conclude the existence and uniqueness it is enough to prove that $a^{\hat\kappa}$ is coercive and continuous. For $ \hat\kappa > 0$, we can  easily see that 
$$a^{{\hat\kappa}}(G,G) \geq \min\{ 1,\hat\kappa V_{1}(\tfrac{1}{2}) \} \left( \Vert G \Vert_{\gamma/2}^{2} + \Vert G \Vert^{2} \right) = \min\{ 1,\hat\kappa V_{1}(\tfrac{1}{2}) \} \Vert G \Vert_{\mathcal{ H}^{\gamma/2}_{0}}^{2}.$$
For $\hat\kappa = 0$, since on $\mathcal{H}^{\gamma}_{0}$ the norms $\Vert \cdot\Vert_{\gamma/2}$ and $\Vert \cdot\Vert_{\mathcal{H}^{\gamma/2}}$ are equivalent we have that
$$a^{{0}}(G,G) = \Vert G\Vert^{2}_{\gamma/2} \gtrsim
 \Vert G \Vert_{\mathcal{ H}^{\gamma/2}_{0}}^{2}.$$
 Therefore $a^{\hat\kappa}$ is coercive for $\hat\kappa \geq 0$. Moreover, by using the Cauchy-Schwarz's inequality we obtain that
$$\vert a^{{\hat\kappa}}(F,G) \vert \leq \Vert F \Vert_{\gamma/2}\Vert G \Vert_{\gamma/2} + \hat\kappa(\Vert F \Vert_{V_{1}}\Vert G \Vert_{V_{1}}).$$
From the fractional Hardy's inequality (see (\ref{norms_related})) we have that 
$$\vert a^{{\hat\kappa}}(F,G) \vert \ls (\hat\kappa +1)(\Vert F \Vert_{\gamma/2} \Vert G \Vert_{\gamma/2}) $$
and since on $ \mathcal{H}^{\gamma/2}_{0}$  the norms $\Vert \cdot  \Vert_{\gamma/2}$ and $\Vert\cdot \Vert_{\mathcal{H}^{\gamma/2}}$ are  equivalent, we conclude that the bilinear form $a^{\hat\kappa} $ is continuous for $\hat\kappa\geq 0$. This end the proof.

\section*{Acknowledgements}
\pat{This work has been supported by the projects EDNHS ANR-14- CE25-0011, LSD ANR-15-CE40-0020-01 of the French National Research Agency (ANR) and of the PHC Pessoa Project 37854WM. B.J.O. thanks Universidad Nacional de Costa Rica for financial support through his Ph.D grant.}

\pat{This project has received funding from the European Research Council (ERC) under  the European Union's Horizon 2020 research and innovative programme (grant agreement   No 715734). }

\pat{This work was finished during the stay of P.G. at Institut Henri Poincar\'e - Centre Emile Borel during the trimester "Stochastic Dynamics Out of Equilibrium". P.G. thanks this institution for hospitality and support.
The authors
thank the Program Pessoa of Cooperation between Portugal and France with reference
406/4/4/2017/S.
}

\appendix 
\section{Computations involving the generator}
\begin{lem}
\label{lem:compA}
For any $x \ne y \in \Lambda_N$, we have 
\begin{enumerate}[i)]
\item $L_N^0 (\eta_x \eta_y) = \eta_x L_N^0 \eta_y + \eta_y L_N^0 \eta_x -  p(y-x) (\eta_y -\eta_x)^2,$
\item $L_N^r (\eta_x \eta_y) = \eta_x L_N^r \eta_y + \eta_y L_N^r \eta_x,$
\item $L_N^\ell (\eta_x \eta_y) = \eta_x L_N^\ell \eta_y + \eta_y L_N^\ell \eta_x.$ 
\end{enumerate}
\begin{proof}
For i) we have, by definition of $L_N^0$, that  
\begin{equation*}
\begin{split}
L_N^0 (\eta_x \eta_y) =& \cfrac{1}{2}\sum _{\bar x,\bar y\in \Lambda_{N}}p(\bar y-\bar x)\left[(\sigma^{\bar x,\bar y}\eta)_{x}(\sigma^{\bar x,\bar y}\eta)_{y}-\eta_{x}\eta_{y}\right]\\
=&\cfrac{1}{2}\sum _{\bar x,\bar y\in \Lambda_{N}}p(\bar y-\bar x)\left[((\sigma^{\bar x,\bar y}\eta)_{x}\eta_{y}-\eta_{x}\eta_{y})+((\sigma^{\bar x,\bar y}\eta)_{y}\eta_{x}-\eta_{x}\eta_{y})+\right. \\
&\;\left. +(\sigma^{\bar x,\bar y}\eta)_{x}(\sigma^{\bar x,\bar y}\eta)_{y}-(\sigma^{\bar x,\bar y}\eta)_{x}\eta_{y}-(\sigma^{\bar x,\bar y}\eta)_{y}\eta_{x}+\eta_{x}\eta_{y}\right]\\
=&\eta_x L_N^0 \eta_y + \eta_y L_N^0 \eta_x + \cfrac{1}{2}\sum _{\bar x,\bar y\in \Lambda_{N}}p(\bar y-\bar x)\left[(\sigma^{\bar x,\bar y}\eta)_{x}-\eta_{x}\right] \left[ (\sigma^{\bar x,\bar y}\eta)_{y}-\eta_{y}\right] \\
=& \eta_x L_N^0 \eta_y + \eta_y L_N^0 \eta_x -  p(y-x) (\eta_y -\eta_x)^2.
\end{split}
\end{equation*}
In order to prove ii), note that $\left[(\sigma^{\bar x}\eta)_{x}-\eta_{x}\right] \left[ (\sigma^{\bar x}\eta)_{y}-\eta_{y}\right]$ is equal to zero, for all $\bar x \in \bZ$. Thus, by definition of $L_N^r$, we have that 
\begin{equation*}
\begin{split}
L_N^r (\eta_x \eta_y) =&\sum_{\bar x\in\Lambda_{N},\bar y\ge N}p(\bar y-\bar x)\left[\eta_{\bar x}(1-\beta)+(1-\eta_{\bar x})\beta\right]\left[(\sigma^{\bar x}\eta)_{x}(\sigma^{\bar x}\eta)_{y}-\eta_{x}\eta_{y}\right]\\
&= \eta_x L_N^r \eta_y + \eta_y L_N^r \eta_x +\\
&\sum_{\bar x\in\Lambda_{N},\bar y\ge N}p(\bar y-\bar x)\left[\eta_{\bar x}(1-\beta)+(1-\eta_{\bar x})\beta\right]\left[(\sigma^{\bar x}\eta)_{x}-\eta_{x}\right] \left[ (\sigma^{\bar x}\eta)_{y}-\eta_{y}\right]\\
&=\eta_x L_N^r \eta_y + \eta_y L_N^r \eta_x.
\end{split}
\end{equation*}
The proof of the third expression is analogous.
\end{proof}
\end{lem}

\

\end{document}